\documentclass[letterpaper]{article}

\usepackage{microtype}
\usepackage{graphicx}
\usepackage{subfigure}
\usepackage{booktabs} 
\usepackage{enumitem}

\usepackage{hyperref}
\usepackage{bm}

\usepackage[accepted]{icml2023}

\usepackage{amsmath}
\usepackage{amssymb}
\usepackage{mathtools}
\usepackage{amsthm}

\usepackage[capitalize,noabbrev]{cleveref}

\theoremstyle{plain}
\newtheorem{theorem}{Theorem}[section]

\newtheorem{lemma}[theorem]{Lemma}
\newtheorem{corollary}[theorem]{Corollary}
\theoremstyle{definition}
\newtheorem{definition}[theorem]{Definition}

\theoremstyle{remark}
\newtheorem{remark}[theorem]{Remark}

\newcommand{\eqn}[1]{(\ref{eqn:#1})}

\newcommand{\thm}[1]{\hyperref[thm:#1]{Theorem~\ref*{thm:#1}}}
\newcommand{\cor}[1]{\hyperref[cor:#1]{Corollary~\ref*{cor:#1}}}
\newcommand{\define}[1]{\hyperref[def:#1]{Definition~\ref*{def:#1}}}
\newcommand{\defn}[1]{\hyperref[defn:#1]{Definition~\ref*{defn:#1}}}
\newcommand{\lem}[1]{\hyperref[lem:#1]{Lemma~\ref*{lem:#1}}}
\newcommand{\prop}[1]{\hyperref[prop:#1]{Proposition~\ref*{prop:#1}}}
\newcommand{\prob}[1]{\hyperref[prob:#1]{Problem~\ref*{prob:#1}}}
\newcommand{\assum}[1]{\hyperref[assum:#1]{Assumption~\ref*{assum:#1}}}
\newcommand{\fig}[1]{\hyperref[fig:#1]{Figure~\ref*{fig:#1}}}
\newcommand{\tab}[1]{\hyperref[tab:#1]{Table~\ref*{tab:#1}}}
\newcommand{\alg}[1]{\hyperref[alg:#1]{Algorithm~\ref*{alg:#1}}}
\renewcommand{\sec}[1]{\hyperref[sec:#1]{Section~\ref*{sec:#1}}}
\newcommand{\de}[1]{\hyperref[def:#1]{Definition~\ref*{def:#1}}}
\newcommand{\append}[1]{\hyperref[append:#1]{Appendix~\ref*{append:#1}}}
\newcommand{\fac}[1]{\hyperref[fac:#1]{Fact~\ref*{fac:#1}}}
\newcommand{\lin}[1]{\hyperref[lin:#1]{Line~\ref*{lin:#1}}}

\newcommand{\E}{\mathbb{E}}

\usepackage[textsize=tiny]{todonotes}

\icmltitlerunning{Near-Optimal Quantum Coreset Construction Algorithms for Clustering}

\DeclareMathOperator{\poly}{poly}
\DeclareMathOperator{\polylog}{polylog}
\DeclareMathOperator{\cost}{cost}

\DeclareMathOperator{\opt}{OPT}
\DeclareMathOperator{\Approx}{Approx}

\DeclareMathOperator{\dist}{dist}
\DeclareMathOperator{\round}{round}

\DeclareMathOperator{\tin}{tiny}
\DeclareMathOperator{\Var}{Var}
\DeclareMathOperator{\close}{close}
\DeclareMathOperator{\far}{far}
\DeclareMathOperator{\Angry}{Angry}
\DeclareMathOperator{\Happy}{Happy}

\usepackage{braket}

\let\epsilon\varepsilon

\begin{document}

\twocolumn[
\icmltitle{Near-Optimal Quantum Coreset Construction Algorithms for Clustering}



\icmlsetsymbol{equal}{*}

\begin{icmlauthorlist}
\icmlauthor{Yecheng Xue}{equal,cfcs}
\icmlauthor{Xiaoyu Chen}{equal,eecs}
\icmlauthor{Tongyang Li}{cfcs}
\icmlauthor{Shaofeng H.-C. Jiang}{cfcs}
\end{icmlauthorlist}

\icmlaffiliation{cfcs}{Center on Frontiers of Computing Studies, Peking University, Beijing, China}
\icmlaffiliation{eecs}{School of Electronics Engineering and Computer Science, Peking University, Beijing, China}

\icmlcorrespondingauthor{Shaofeng H.-C. Jiang}{shaofeng.jiang@pku.edu.cn}
\icmlcorrespondingauthor{Tongyang Li}{tongyangli@pku.edu.cn}

\icmlkeywords{Machine Learning, ICML}

\vskip 0.3in
]



\printAffiliationsAndNotice{\icmlEqualContribution} 

\begin{abstract}
    $k$-Clustering in $\mathbb{R}^d$ (e.g., $k$-median and $k$-means) is a fundamental machine learning problem. While near-linear time approximation algorithms were known in the classical setting for a dataset with cardinality $n$, it remains open to find sublinear-time quantum algorithms.
    We give quantum algorithms that find coresets for $k$-clustering in $\mathbb{R}^d$ with $\tilde{O}(\sqrt{nk}d^{3/2})$ query complexity.
    Our coreset reduces the input size from $n$ to $\mathrm{poly}(k\epsilon^{-1}d)$, so that existing $\alpha$-approximation algorithms for clustering can run on top of it and yield $(1 + \epsilon)\alpha$-approximation.
    This eventually yields a quadratic speedup for various $k$-clustering approximation algorithms.
    We complement our algorithm with a nearly matching lower bound,
    that any quantum algorithm must make $\Omega(\sqrt{nk})$ queries in order to achieve even $O(1)$-approximation for $k$-clustering.
\end{abstract}

\section{Introduction}

Clustering is a fundamental machine learning task that has been extensively studied in areas including computer science and operations research.
A typical clustering problem is $k$-median clustering in $\mathbb{R}^d$.
In $k$-median clustering, we are given a set of data points $D \subseteq \mathbb{R}^d$ and an integer parameter $k$, and the goal is to find a set $C \subset \mathbb{R}^d$ of $k$ points, called the center set, such that the following objective is minimized:
\begin{equation}
    \label{eqn:cost}
    \cost(D, C) := \sum_{x \in D}{\dist(x, C)},
\end{equation}
here $\dist(x, y) := \|x - y\|_2$, $\dist(x, C) := \min\limits_{c \in C}{\dist(x, c)}$.

In the classical setting, 
$k$-median clustering is shown to be NP-hard~\cite{DBLP:journals/siamcomp/MegiddoS84} even in the planar case (i.e., Euclidean $\mathbb{R}^2$),
and polynomial time approximation algorithms were the focus of study.
Furthermore, even allowing $O(1)$-approximation,
a fundamental barrier is still that the algorithm must make $\Omega(nk)$ accesses to the distances between the $n$ input points~\cite{MettuP04}.

In this paper, we study quantum algorithm and complexity for $k$-median clustering (and more generally, $(k, z)$-clustering as in \de{kzc}).
This is motivated by quantum algorithms for various related data analysis problems, such as classification~\cite{kapoor2016quantum,li2019sublinear,li2021sublinear}, nearest neighbor search~\cite{wiebe2015quantum}, support vector machine~\cite{rebentrost2014QSVM}, etc. Many of these quantum algorithms are originated from the Grover algorithm~\cite{grover1996fast}, which can find an item in a data set of cardinality $n$ in time $O(\sqrt{n})$, a quadratic quantum speedup compared to the classical counterpart.
Hence, a natural question is whether quantum algorithms can break the aforementioned classical $\Omega(nk)$ lower bound for $k$-median clustering, with a practical goal of achieving complexity $O(\sqrt{nk})$, while still achieving $O(1)$ or even $(1 + \epsilon)$-approximation.

\paragraph{Coresets}
To this end, we consider constructing coresets~\citep{Har-PeledM04,feldman2011unified} by quantum algorithms.  
The coreset is a powerful technique for dealing with clustering problems.
Roughly speaking, an $\epsilon$-coreset is a tiny proxy of the potentially huge data set,
such that the clustering cost on any center set is preserved within $\epsilon$ relative error.
For $k$-median, an $\epsilon$-coreset of size $\poly(k\epsilon^{-1})$ has been known to exist (see e.g., \citealp{DBLP:conf/focs/SohlerW18,HuangV20,BravermanJKW21,cohen2021new,Cohen-AddadLSS22,SSLCS22}),
which is independent of both the dimension $d$ and size $n$ of the data set.
Once such a coreset is constructed, one can approximate $k$-means efficiently by plugging in existing approximation algorithms (so that the input size is reduced to the size of the coreset which is only $\poly(k \epsilon^{-1})$).
In addition, coresets can also be applied to clustering algorithms in sublinear settings
such as streaming~\citep{Har-PeledM04}, distributed computing~\citep{BalcanEL13}, and dynamic algorithms~\citep{HenzingerK20}.

\paragraph{Contributions}
We propose the first quantum algorithm for constructing coresets for $k$-median
that runs in $\tilde{O}\left(\sqrt{nk}\right)$ time,\footnote{In this paper, we use $\tilde{O}$ to omit poly-logarithmic terms in $O$.} which breaks the fundamental linear barrier of classical algorithms.

\begin{theorem}[Informal version of \thm{coreset}]
\label{thm:intro_ub}
There exists an quantum algorithm that given $\varepsilon > 0$ and an $n$-point data set $D \subset \mathbb{R}^d$, returns an $\varepsilon$-coreset of size $\tilde{O}\left( kd\polylog(n)/\varepsilon^{2} \right)$ for $k$-median over $D$ with success probability at least $2/3$, with query complexity $\tilde{O}\left(\sqrt{nk}d^{3/2}/\varepsilon \right)$ and additional $\poly\left(kd\log n/\varepsilon\right)$ processing time.
\end{theorem}

Our coreset construction also yields coresets of a similar size for the related $k$-median clustering problem (and more generally, ($k$, $z$)-clustering, see \define{kzc}), using the same order of query and processing time (see \thm{coreset}).
The size bound of our coreset stated in \thm{intro_ub} may not be optimal,
but since it already has size $\poly(k\epsilon^{-1}d)$,
one can trivially apply the state-of-the-art classical coreset construction algorithm on top of our coreset to obtain improved bounds. For instance, we can obtain coresets for $k$-means of size $O\left(k\epsilon^{-4}\right)$ using~\citet{Cohen-AddadLSS22},
and an alternative size bound of $O\left(k^{1.5}\epsilon^{-2}\right)$ using~\citet{SSLCS22}.
These require the same order of query and processing time as in \thm{intro_ub}.

In addition, by a similar argument, our coreset readily implies efficient approximation algorithms for clustering problems.
In particular, one constructs an $\epsilon$-coreset $S$ as in \thm{intro_ub} and
applies an existing $\alpha$-approximate algorithm with input $S$, then it yields an $O((1 + \epsilon) \alpha)$-approximation to the original problem.
The query complexity of the entire process
remains the same as in \thm{intro_ub},
and it only incurs additional $T(\poly(k \epsilon^{-1}\log n))$ processing time, provided that the approximation algorithm runs in $T(n)$ time for an $n$-point dataset.
This particularly implies a quantum PTAS (for fixed $d$) for $k$-median and $k$-means using $\tilde{O}\left(\sqrt{nk}d^{3/2}/\varepsilon\right)$ queries,
and $\poly(k) \cdot f(d, \epsilon) $ processing time for some function $f$ of $d$ and $\epsilon$, by applying~\citet{DBLP:journals/jacm/Cohen-AddadFS21}.

Our quantum algorithms (and lower bounds) for clustering are summarized in \tab{main}.

\begin{table*}[!htp]
  \begin{center}
  \caption{Classical and quantum complexity bounds for clustering in $\mathbb{R}^d$ (omitting dependence in $d$). Note that an $O(1)$-coreset for $k$-median (resp. $k$-means) implies an $O(1)$-approximate solution for $k$-median (resp. $k$-means).}\label{tab:main}
  \begin{tabular}{lllll}
    \toprule
    Reference & Type & Problem & Time Complexity \\
    \midrule
    \citet{BF17} & Classical upper bound  & $\epsilon$-coreset for $k$-median &  $\tilde{O}(n+\poly(k))$ \\
    \thm{coreset} & Quantum algorithm & $\epsilon$-coreset for $k$-median & $\tilde{O}(\sqrt{nk}/\varepsilon)$ \\
    \thm{coreset} & Quantum algorithm & $\epsilon$-coreset for $k$-means & $\tilde{O}(\sqrt{nk}/\varepsilon)$  \\
    \thm{lower-clustering-main} & Quantum lower bound & $O(1)$-approximate $k$-median  & $\Omega(\sqrt{nk})$ \\
    \thm{lower-clustering-main} & Quantum lower bound & $O(1)$-approximate $k$-means  & $\Omega(\sqrt{nk})$ \\
    \bottomrule
  \end{tabular}
  \end{center}
\vspace{-2mm}
\end{table*}

\paragraph{Techniques}
The general idea of our algorithm is to quantize and combine two existing algorithms, the approximate algorithm by~\cite{thorup2001quick} and a recent seminal coreset construction algorithm by~\citet{cohen2021new}. 

In a high level, we start with computing a bicriteria solution which uses slightly more than $k$ points to achieve $O(1)$-approximation to $\opt$. This step is based on~\cite{thorup2001quick}. Then given this solution, we partition the dataset $D$ into groups, perform a sampling procedure in each group, and re-weight the sampled points to form the coreset, following the idea in~\cite{cohen2021new}.

For the bicriteria approximation, we provide a quantum implementation for the algorithm by~\citet{thorup2001quick} in \sec{bicriteria} to obtain a solution that has $O(k \poly\log n)$ points with cost being $O(1)$ multiple of $\opt$. A key step in this algorithm is to query for the nearest neighbor of each data point $x \in D$ in a given point set with size $O(k)$. This step is straightforward in the classical setting with cost $O(nk)$ by calculating the exact nearest neighbor and store the results. However, to achieve the $\tilde{O}\left(\sqrt{nk}\right)$ complexity in the quantum setting, we need improvements on nearest neighbor search. We make use of a well-known approximate nearest neighbor search technique called \emph{locality sensitive hashing} (LSH), and the version that we use gives $2(1+\varepsilon)$-approximate nearest neighbor using $N^{\poly(1/\varepsilon)}$ preprocessing time and $\tilde{O}(d)\cdot\poly(1/\varepsilon)$ query time for $N$ points in $\mathbb{R}^d$~\cite{indyk1998approximate}.
In \sec{coreset_details}, we also use our quantum implementation of LSH (\lem{qANN}) to construct a unitary that encodes the clusters induced by the approximate solution, i.e., maps each $x \in D$ to a corresponding center.

After we obtain a bicriteria approximation $A$ in the previous step,
we adapt the algorithm in~\cite{cohen2021new} to build the coreset, where the dataset $D$ is partitioned into $\tilde{O}(z^2/\varepsilon^2)$ groups with respect to $A$.
In this partition procedure, a key step is to calculate $\cost(C_i,A) = \sum_{x\in C_i}\dist^z(x,A)$ for each cluster $C_i$ induced by $A$.
While this seems simple in the classical setting where one directly computes the cost of each data point and summing up over each cluster in $O(nk)$ time,
this task is nontrivial in quantum since we aim for sublinear complexity.
To design a sublinear quantum algorithm that approximately computes the cost of all clusters simultaneously, 
we propose a new subroutine called multidimensional quantum approximate summation (\thm{mqsum}).
Specifically, given an oracle $O:\ket{i}\ket{0}\ket{0}\rightarrow \ket{i}\ket{\tau(i)}\ket{f(i)}$, where $\tau\colon[n]\rightarrow[m]$ is a partition and $f\colon[n]\rightarrow \mathbb{R}_{\geq 0}$ is a bounded function, \thm{mqsum} shows that using in total $\tilde{O}\left(\sqrt{nm/\varepsilon}\right)$ queries to $O_{\tau}$ one can obtain $\varepsilon$-estimation of $\sum_{\tau(i) = j}f(i)$ for each $j \in [m]$. This algorithm may be of independent interest.

To complement our algorithm results, we also prove quantum lower bounds for approximate $k$-means clustering. We consider three settings in which an $\varepsilon$-coreset, an $\varepsilon$-optimal set of centers, or an $\varepsilon$-estimate of the optimal clustering cost are outputted, and we prove $\Omega\left(\sqrt{nk}\varepsilon^{-1/2}\right),\Omega\left(\sqrt{nk}\varepsilon^{-1/6}\right)$ and $\Omega\left(\sqrt{nk}+\sqrt{n}\varepsilon^{-1/2}\right)$ lower bounds, respectively. These quantum lower bounds confirm that our quantum algorithms for clustering problems are \emph{near-optimal in $n$ and $k$}, up to a logarithmic factor.
In general, we start with proving the quantum lower bounds when $k=1$ by reducing from the approximate quantum counting problem~\cite{NayakW99}. We then obtain the general bounds with $\sqrt k$ factors by applying composition theorems (\thm{composition_theorem}) in a refined manner. We speculate that the gap between different settings might be intrinsic, which is also discussed in a recent paper \cite{CharikarW22}. 

\paragraph{Related work}
In general, quantum algorithms for machine learning are of general interest~\cite{biamonte2017quantum,schuld2018supervised,dunjko2018machine}. 
We compare our results to existing literature in quantum machine learning as follows.

\citet{aimeur2007quantum} conducted an early study on quantum algorithms for clustering, including divisive clustering, $k$-median clustering, and neighborhood graph construction. Their $k$-median algorithm has complexity $O\left(n^{3/2}/\sqrt{k}\right)$, which is at least $n$ and slower than our quantum algorithm.

\citet{lloyd2013quantum} gave a quantum algorithm for cluster assignment and cluster finding with complexity $\poly(\log nd)$. However, their quantum algorithm requires the input data to be sparse with efficient access to nonzero elements, i.e., each of $x_{1},\ldots,x_{n}$ has $\poly(\log d)$ nonzero elements and we can access these coordinates in $\poly(\log d)$ time. In addition, their algorithm outputs quantum states instead of classical vectors. More caveats are listed in~\citet{aaronson2015read}.

The most relevant result is~\citet{kerenidis2019qmeans}, which gave an quantum algorithm named q-means for $k$-means clustering. Q-means has complexity $\tilde{O}(k^{2}d\eta^{2.5}\epsilon^{-3}+k^{2.5}\eta^{2}\epsilon^{-3})$ per iteration for well-clusterable datasets, where $\eta$ is a scaling factor for input data such that $1\leq\|x_{i}\|_{2}^{2}\leq\eta$ for all $i\in[n]$. For general datasets the complexity is larger and depends on condition number parameters. Q-means can be extended to spectral clustering~\cite{kerenidis2021quantum}. As a comparison, our quantum algorithm does not require the well-clusterable assumption (nor condition numbers related to this) and can be regarded as a direct speedup of common classical algorithms for $k$-means. In addition, our quantum algorithm only has $\poly(\log\eta)$ dependence, and the dependence on $k$ and $1/\epsilon$ is also better.

There are also heuristic quantum machine learning approaches for clustering~\cite{rigetti2017unsupervised,poggiali2022quantum} and other problems in data analysis~\cite{schuld2017implementing,farhi2018classification,KL18,havlivcek2019supervised}. These results do not have theoretical guarantees at the moment, and we look forward to their further developments on heuristic performances and provable guarantees. In addition,~\citet{chia2022sampling} proposed a quantum-inspired classical algorithm for 1-mean clustering with sampling access to input data.

\paragraph{Open questions}
Our work leaves several natural open questions for future investigation:
\begin{itemize}[leftmargin=*]
\item Can we give fast quantum coreset construction algorithms for other related clustering problems with complexity $\tilde{O}\left(\sqrt{nk}\right)$? Potential problems include fair clustering~\cite{Chierichetti0LV17}, capacitated clustering~\cite{cohenaddad2019capacitated,braverman2022power}, $k$-center clustering~\cite{agarwal2002exact}, etc.

\item Can we give fast quantum coreset construction algorithms for clustering problems in more general metric spaces? Note that~\citet{cohen2021new} gives the result for various metric space, such as doubling metrics, graphs, and general discrete metric spaces, while still achieving $\tilde{O}(nk)$ time in the classical setting.
However, to achieve this similar coreset size bound using time $O(\sqrt{nk})$ in the quantum setting seems nontrivial; for instance, one cannot make use of the LSH technique that we use to speed up the approximate nearest neighbor search.

\end{itemize}

\section{Preliminaries}
\subsection{Notations}
We give the notations and definitions used in the following. In this paper, we focus on the \emph{($k$, $z$)-clustering} problem in an Euclidean space:
\begin{definition}[($k$, $z$)-Clustering]\label{def:kzc}
Given a data set $D \subset \mathbb{R}^d$, $z \geq 1$ and integer $k \geq 1$, the ($k$, $z$)-clustering problem is to find a set $C \subset \mathbb{R}^d$ with size $k$, minimizing the cost function
$$
\cost(D,C) := \sum_{x \in D}\left(\dist(x,C)\right)^z.
$$
Here $\dist(x,y) := ||x-y||_2$, $\dist(x,C):=\min\limits_{c \in C}\dist(x,c)$.
\end{definition}
For $z= 1$ it is also known as the \emph{$k$-median} problem, and for $z = 2$ it is also known as the \emph{$k$-means} problem. Any set $C \subset \mathbb{R}^d$ with size $k$ can be seen as a \emph{solution}. Given a solution $C$, a point $c\in C$ is a \emph{center}. We can map each data point $x \in D$ to a certain $c \in C$, and the set $\{x\in D: x \mbox{ is mapped to } c\}$ is a \emph{cluster}. Let $\opt$ be the cost of the optimal solution. We assume the distance is rescaled so that the minimum intra-point distance is $1$, and we assume the diameter of the point set is $\poly(n)$. Hence, $\opt = \poly(n)$.

An important related concept is the \emph{coreset}:
\begin{definition}[Coreset]\label{def:coreset}
Given a data set $D \subset \mathbb{R}^d$,  $z \geq 1$ and integer $k \geq 1$,
a weighted set $S$ with weight function $w \colon S \to \mathbb{R}_+$ is called an $\epsilon$-coreset if 
\[
    \forall C \subset \mathbb{R}^d, |C| \leq k, \ \ 
    \cost(S,C) \in (1 \pm \epsilon) \cdot \cost(D,C)
\]
where $\cost(S,C) = \sum_{s\in S}w(s)\cdot \dist^z(s,C)$. 
\end{definition}
As a special case, we call $S$ unweighted if $w(s) = 1\ \forall s\in S$.

In this paper, we refer to an $\varepsilon$-estimation for a number $p$ by $\tilde{p}$ if $|\tilde{p}-p| \leq \varepsilon p$. We use $[n]$ for $\{1,\ldots,n\}$.

\subsection{Basics of Quantum Computing}
The basic unit of a classical computer is a bit, and in quantum computing it is a \emph{qubit}. Mathematically, a system of $m$ qubits forms an $M$-dimensional Hilbert space for $M = 2^m$. Any \emph{quantum state} $\ket{\phi}$ in this space can be written as
\begin{equation}\label{eqn:quantum_state}
\ket{\phi} = \sum_{i = 0}^{M-1} \alpha_i\ket{i},\ \text{ where }\sum_{i=0}^{M-1}|\alpha_i|^{2}=1.
\end{equation}
Here $\{\ket{0},\ldots,\ket{M-1}$\} forms an orthonormal basis in the Hilbert space called as the \emph{computational basis}, and $\alpha_i \in \mathbb{C}$ is called as the \emph{amplitude} of $\ket{i}$. Intuitively, the quantum state $\ket{i}$ can be regarded as a classical state $i$, and the quantum state $\ket{\phi}$ in \eqn{quantum_state} is a \emph{superposition} of classical states.
The operations in quantum computing are \emph{unitaries} matrices following the principles of linear algebra. Specifically, a unitary acting on an $M$-dimensional Hilbert space can be formulated as follows:
$$
U_f\colon \ket{i}\ket{0}\rightarrow \ket{i}\ket{f(i)}, \ \forall i \in \{0,\ldots,M-1\}.
$$
Note that due to linearity, $U_f$ works not only for the basis vectors $\{\ket{i}\}_{i = 0}^{M-1}$, but also for any quantum state in this Hilbert space. For example, by applying $U_f$ to $\ket{\phi}$ we obtain the following quantum state
$$
\ket{\phi} = \sum_{i = 0}^{M-1} \alpha_i\ket{i} \ \xrightarrow{U_f} \sum_{i = 0}^{M
-1}\alpha_i \ket{f(i)}.
$$
This allows us to perform calculations ``in parallel" and achieve the quantum speedup.

Quantum access to the input data is also unitary and can be encoded as a \emph{quantum oracle}. We state the definitions of the \emph{probability oracle} and the \emph{binary oracle} as follows:

\begin{definition}[Probability Oracle]\label{def:quantum_probability_oracle}
Let $p\colon [M]\rightarrow \mathbb{R}_{\geq 0}$ be a probability distribution. We say $O_p$ is a probability oracle for $p$ if
$$
O_p\colon \ket{0}\rightarrow \sum_{j \in [M]}\sqrt{p(j)}\ket{j}\ket{\phi_j},
$$
where $\ket{\phi_j}$ are arbitrary $\ell_{2}$-normalized vectors.
\end{definition}

\begin{definition}[Binary Oracle]\label{def:quantum_binary_oracle}
Let $D = \{x_1,\ldots,x_n\}$ be a subset of $\mathbb{R}^d$. We say $O_D$ is a binary oracle for $D$ if
$$
O_D\colon \ket{i}\ket{0}\rightarrow\ket{i}\ket{x_i}, \ \forall i \in [n].
$$
\end{definition}
The definition of the binary oracle also fits for any vector $w = (w_1,\ldots,w_N) \in \mathbb{R}^N$. Binary oracle is a common input model in quantum algorithms and we also call the binary oracle of $D$ (or $w$) as the quantum query to $D$ (or $w$). Besides, for two point sets $S\subset D \subset \mathbb{R}^d$, we say $O_S$ is the membership query to $S$ if
$$
O_S\colon \ket{x}\ket{0}\rightarrow\ket{x}\ket{I(x\in S)}, \ \forall x \in D
$$
where $I(x\in S)$ is the indicator for whether $x \in S$.

In a quantum algorithm, we can also write information to a \emph{quantum-readable classical-writable classical memory} (QRAM) and make it encoded as an oracle~\cite{giovannetti2008quantum}. We refer \emph{query complexity} as the number of queries to the input oracle and the QRAM. \emph{Time complexity} is referred as the total processing time, including all the use of queries, quantum gates, and classical operations.

\subsection{Quantum Speedup}
Here, we introduce basic problems which can be sped-up by quantum computing. Those tools are rudimentary to our quantum algorithms as well as others in machine learning.

\paragraph{Quantum Sampling}
In our quantum algorithms, we the following quantum sampling algorithm:
\begin{lemma}[Rephrased from Theorem 1 of~\citealt{hamoudi2022preparing}]\label{lem:qsampling}
 There is a quantum algorithm such that: given two integers $1 \leq m \leq n$, a real $\delta > 0$, and a non-zero vector $w \in \mathbb{R}_{\geq 0}^n$, with a probability at least $1-\delta$, the algorithm outputs a sample set $S$ of size $m$ such that each element $i \in [n]$ is sampled with probability proportional to $w_i$ using $O(\sqrt{nm}\log(1/\delta))$ quantum queries to $w$ in expectation. 
\end{lemma}

\paragraph{Quantum Counting and Search}
In quantum computing, counting the number of points satisfying a specific property can be solved with quadratic speedup:
\begin{lemma}[Theorem 15 of \citealt{brassard1998quantum}]\label{lem:counting} 
There is a quantum algorithm such that given a real $\delta>0$ and two sets $S \subset D \subset \mathbb{R}^d$, $|S| = m$, $|D| = n$, it outputs an $\varepsilon$-estimation $\tilde{m}$ for $m$ with probability at least $1-\delta$ using $O(\varepsilon^{-1}\sqrt{n/m}\log(1/\delta))$ queries to $D$ and membership queries to $S$.
\end{lemma}

Furthermore, quantum speedup can also be achieved with outputting all such points, known as repeated Grover search:
\begin{lemma}[Claim 2 of~\citealt{apers2020quantum}]\label{lem:getall}
There is a quantum algorithm such that given a real $\delta>0$ and two sets $S \subset D \subset \mathbb{R}^d$, $|S| = m$, $|D| = n$, it finds $S$ with probability at least $1-\delta$ using $\tilde{O}(\sqrt{nm}\log(1/\delta))$ queries to $D$ and membership queries to $S$. 
\end{lemma}

\paragraph{Quantum Sum Estimation}
Beyond counting and search, quadratic quantum speedup can also be achieved for estimating the sum of a set of numbers:
\begin{lemma}[Rephrased from Lemma 6 of~\citealt{li2019sublinear}]\label{lem:qsumestimation}
Consider $D = \{x_1,\ldots,x_n\} \subset \mathbb{R}^n_{\geq 0}$ and denote  $x = \sum_{i = 1}^N x_i$ as the sum of all the elements in $D$. There is a quantum algorithm such that given $\delta > 0$, it outputs $\tilde{x}$ as an $\varepsilon$-estimation for $x$ with probability at least $1 -\delta$, using $O(\sqrt{n}\log(1/\delta)/\varepsilon)$ queries to $D$. 
\end{lemma}

\section{Coreset Construction}

This section presents a quantum algorithm for coreset construction in $\tilde{O}(\sqrt{nk})$ time. This algorithm combines and quantizes two existing classical algorithms, the bicriteria approximation algorithm of~\citet{thorup2001quick} and the coreset construction algorithm (based on an approximate solution) of~\citet{cohen2021new}. This paper focuses on the ($k$, $z$)-clustering problem (\de{kzc}) over a size-$n$ data set $D = \{x_1,\ldots,x_n\} \subset \mathbb{R}^d$, and assumes the access to oracle $O_D\colon \ket{i}\ket{0}\rightarrow\ket{i}\ket{x_i}$ $\forall i \in [n]$. The main result for coreset construction is as follows.

\begin{theorem}\label{thm:coreset}
There exists a quantum algorithm such that given
data set $D \subset \mathbb{R}^d$, positive real $\epsilon < 1/2^{O(z)}$, $z \geq 1$, and integer $k \geq 1$,
it returns an $\epsilon$-coreset for ($k$, $z$)-clustering over $D$ of size $\tilde{O}\left(2^{O(z)}kd\polylog(n)\max{(\varepsilon^{-2},\varepsilon^{-z})}\right)$ with success probability at least $2/3$ using:
\begin{itemize}[leftmargin=*]
\item $\tilde{O}\left(2^{O(z)}\sqrt{nkd}\max{(\varepsilon^{-1},\varepsilon^{-z/2})}\right)$ queries to $O_D$,
\item $\tilde{O}\left(2^{O(z)}\sqrt{nk}d^{3/2}\max{(\varepsilon^{-1},\varepsilon^{-z/2})}\right)$ queries to QRAM,
\item $\poly\left(kd\log n/\varepsilon^z\right)$ additional processing time.
\end{itemize}
\end{theorem}

\begin{remark}
\label{remakr:jl}
    When $d \geq \Omega(\log n / \epsilon^2)$,
    one can apply the Johnson-Lindenstrauss transform~\cite{JL84} as a preprocessing step to obtain the following alternative bounds.
    \begin{itemize}[leftmargin=*]
        \item $\tilde{O}\left(2^{O(z)}k\polylog(n)\max{(\varepsilon^{-4},\varepsilon^{-z-2})}\right)$ coreset size,
        \item $\tilde{O}\left(2^{O(z)}\sqrt{nk}\max{(\varepsilon^{-2},\varepsilon^{-z/2-1})}\right)$ queries to $O_D$,
        \item $\tilde{O}\left(2^{O(z)}\sqrt{nk}d\max{(\varepsilon^{-4},\varepsilon^{-z/2-3})}\right)$ queries to QRAM, 
        \item $\poly\left(k\log n/\varepsilon^z\right) + O\left(d\log n/\varepsilon^2\right)$ additional processing time.
    \end{itemize}
    These bounds have tight asymptotic dependence in $d$.
\end{remark}

The quantum algorithm contains two parts. First, \sec{bicriteria} presents an algorithm to compute a bicriteria approximate solution $A$, which is an approximate solution with size slightly larger than $k$. Then, based on this $A$, \sec{coreset_details} presents an algorithm for coreset construction. Combining the two algorithms directly yields \thm{coreset}.

For clustering problem it is a basic subroutine to find the nearest neighbor of each data point $x \in D$ in a given set $A \subset \mathbb{R}^d$ since the optimization objective is the cost function $\cost(D,A) = \sum_{x \in D}\dist^z(x,A)$ for any $A \subset \mathbb{R}^d$. It always holds that $|A| = k\polylog(n)$. This can be easily implemented in the classical setting, since the nearest neighbor of all the $x \in D$ can be found with $\tilde{O}(nk)$ time and all the information can be stored with $O(n)$ space. However, this approach cannot be easily adapted to yield the $\tilde{O}(nk)$ complexity in the quantum setting. In particular, the step of exactly computing the nearest center can require $\Omega(k)$ time. Hence, this paper uses a mapping that maps each point $x \in D$ to an approximately nearest neighbor in $A$ instead. This paper takes the advantage of an existing classical result, which is based on a widely adopted technique, Locality Sensitive Hashing (LSH).

\begin{lemma}[Approximate Nearest Neighbor Search, Rephrased from Theorem 2.10 and Theorem 3.17 of~\citealt{indyk1998approximate}]\label{lem:cANN}
There exists an algorithm such that given two parameters $\delta' > 0$, $\varepsilon \in (0,1/2)$, for any set $A \subset \mathbb{R}^d$, $|A| = m$, it constructs a data structure using $m^{O(\log (1/\varepsilon)/\varepsilon^2)}\log(1/\delta')$ space and preprocessing time, such that for any query $x \in \mathbb{R}^d$, with probability at least $1-\delta'$ it answers $a \in A$ which satisfies $\dist(x,a) \leq 2(1+\varepsilon)\dist(x,A)$ using $\varepsilon^{-2}d\polylog(m/\delta')$ query time.
\end{lemma}

For a quantum version, there exists an algorithm that performs the preprocessing classical and stores the data structure in QRAM~\citep{giovannetti2008quantum}, and then answers queries in a quantum manner based on the stored information. This yields a quantum algorithm with the same query time since any classical operation can be simulated by constant quantum operations. Let $\varepsilon = c/2-1$. Setting $\delta' = \delta/n$ and using the union bound yields the following lemma.

\begin{lemma}[Quantum Approximate Nearest Neighbor Search]\label{lem:qANN}
There exists an algorithm such that given two parameters $\delta > 0$, $c_{\tau} \in [5/2,3)$, for any $A \subset \mathbb{R}^d$, $|A| = m$, it constructs an oracle
$$
O_{\tau}\colon\ket{i}\ket{0}\rightarrow\ket{i}\ket{\tau(i)}, \forall i \in [n]
$$
using $\poly(m\log(n/\delta))$ classical preprocessing time and QRAM space. With success probability at least $1-\delta$, $\tau \colon [n] \rightarrow [m]$ is a mapping such that
$$
\dist(x_i,a_{\tau(i)}) \leq c_{\tau}\dist(x_i,A) \quad \forall i \in [n].
$$
Each query to $O_{\tau}$ uses $d\polylog(mn/\delta)$ queries to QRAM. 
\end{lemma}

\begin{remark}
By the same method, one can construct the oracle to a mapping that maps any $i \in [n]$ to the corresponding center $a_{\tau(i)} \in A$ instead of the index $\tau(i)$, with the same complexity.
\end{remark}

\subsection{Bicriteria Approximation}\label{sec:bicriteria}

For a ($k$, $z$)-clustering problem, its bicriteria approximate solution is defined as follows:
\begin{definition}
\label{def:bicriteria}
Assume that $\opt$ is the optimal cost for the ($k$, $z$)-clustering problem. An ($\alpha$, $\beta$)-bicriteria approximate solution is a point set $A \subset \mathbb{R}^d$ such that $|A| \leq \alpha k$, $\cost(D,A) \leq \beta \opt$.
\end{definition}

This section presents a quantum algorithm that finds a bicriteria solution with $\tilde{O}(d\sqrt{nk})$ query complexity, as stated in \lem{bicriteria} below. This algorithm is a quantization of a classical algorithm by~\citet{thorup2001quick}.

\begin{lemma}\label{lem:bicriteria}
 \alg{bicriteria} outputs an ($O(\log^2 n)$, $2^{O(z)}$)-bicriteria approximate solution $A$ with probability at least $5/6$,
 using $\tilde{O}(\sqrt{nk})$ calls to $O_D$ and its inverse, $\tilde{O}(d\sqrt{nk})$ queries to a QRAM, and $\poly(k\log n)$ additional processing time.
\end{lemma}

\begin{algorithm}[tb]
\caption{Bicriteria Approximation}\label{alg:bicriteria}
\begin{algorithmic}[1]
\STATE{\bfseries input:} $k$, $z$, $n$, oracle $O_D$
\STATE{\bfseries output:} ($O(\log^2 n)$, $2^{O(z)}$)-bicriteria approximation $A$
\STATE\label{lin:3} initialize $t \leftarrow 0$, $D_0 \leftarrow D$, $A_0 \leftarrow \emptyset$, $\tilde{r}_0 \leftarrow n$ 
\REPEAT\label{lin:4}
\STATE\label{lin:5} draw a uniform sample $S_t$ of size $13k\lceil\log n \rceil$ over $D_t$ using \lem{qsampling}. $A_{t+1} \leftarrow S_t\cup A_{t}$
\STATE draw a sample $s_t$ uniformly at random over $D_t$ using \lem{qsampling}
\STATE construct a map $\tau_{t+1}\colon D\rightarrow A_{t+1}$ \\
$d_{t+1} \leftarrow \dist(s_t,\tau_{t+1}(s_t))$\\
$D_{t+1} \leftarrow D_t\setminus\{x\in D: \dist(x,\tau_{t+1}(x)) \leq d_{t+1}\}$
\STATE make an $1/2$-estimation $\tilde{r}_{t+1}$ for $r_{t+1} = |D_{t+1}|$ using \lem{counting}
\STATE $t \leftarrow t+1$
\UNTIL{$\tilde{r}_t \leq 39k\lceil\log n\rceil$ or $t \geq 3\lceil \log n \rceil$} 
\STATE find all points in $D_t$ using \lem{getall}, $A \leftarrow D_t \cup A_t$
\STATE repeat all the steps above for three times and union all the $A$
\end{algorithmic}
\end{algorithm}

$\tau_{t+1}\colon D \rightarrow A_{t+1}$ in \alg{bicriteria} Line 7 is a mapping such that for some constant $c_{\tau}$
$$
\dist(x,\tau_{t+1}(x)) \leq c_{\tau}\dist(x,A_{t+1}).
$$
It holds that $|A_{t+1}| = O(k\polylog n)$ for any $t$. Using \lem{qANN}, for any $c_{\tau} \in [5/2,3)$, an oracle
$$
O_{\tau_{t+1}}\colon\ket{x}\ket{0}\rightarrow\ket{x}\ket{\tau(x)}
$$
can be constructed using $\poly(k\log(n))$ classical preprocessing time and QRAM space, and each query to $O_{\tau_{t+1}}$ uses $d\polylog(kn)$ queries to QRAM and constant query to $O_D$.

\begin{proof}[Proof of \lem{bicriteria}]
\alg{bicriteria} is a quantum implementation of Algorithm D of~\citet{thorup2001quick}, which constructs a set of size $O(k\log^2n/\varepsilon)$ that contains a factor $2+\varepsilon$ approximation to $k$-median problem for $\varepsilon\in (0,1/2)$ with probability at least $1/2$. \alg{bicriteria} has small difference from the classical algorithm but it does not influence the correctness; a detailed proof is given in \append{proof_biapprox}.

In the classical setting, to identify the set $D_t$ one can list all the points in it or in the dataset $D$ make those points marked. In the quantum setting, to identify $D_t$ is to construct the unitary
$$
U_{D_t}\colon \ket{x}\ket{0}\rightarrow \ket{x}\ket{I(x \in D_t)}, \quad \forall x \in D,
$$
where $I(x \in D_{t})$ is the indicator for whether $x \in D_{t}$. This unitary can be constructed iteratively:
$$
\begin{aligned}
U_{D_{t+1}}\colon &\ket{x}\ket{0}\ket{0}\ket{0} \\
&\xmapsto{U_{D_t}, O_{\tau_{t+1}}} \ket{x}\ket{I(x \in D_t)}\ket{\tau_{t+1}(x)}\ket{0} \\
&\mapsto \ket{x}\ket{I(x\in D_t)}\ket{\tau_{t+1}(x)}\ket{I(x\in D_{t+1})} \\
&\xmapsto{O_{\tau_{t+1}}^{-1},O_{D_t}^{-1}}  \ket{x}\ket{0}\ket{0}\ket{I(x\in D_{t+1})}.
\end{aligned}
$$
For Line 5-6, \alg{bicriteria} applies \lem{qsampling} with unitary $U_{D_t}$ and $m = 13k\lceil\log n\rceil +1$, which uses $\tilde{O}(\sqrt{nk})$ calls for $O_{D_t}$, $O_D$, and their inverses. This algorithm uses $\tilde{O}(\sqrt{n})$ for Line 8 and $\tilde{O}(nk)$ for Line 11 queries to $O_{D_t}$, $O_D$, and their inverses. All the steps above are repeated for no more than $O(\log n)$ times, and it can be concluded that in total the algorithm uses $\tilde{O}(d\sqrt{nk})$ queries for a QRAM, $\tilde{O}(\sqrt{nk})$ queries for $O_D$ and its inverse, and additional $\poly(k\log n)$ processing time.

We further note that the failure probability of our quantum algorithm gives at most a poly-logarithmic factor. In \alg{bicriteria}, each subroutine in use suffers only a $\log(1/\delta)$ factor to reach a success probability at least $1-\delta$ and each of them is applied no more than $\poly(nk)$ times, so setting the failure probability as $\delta = O(1/\poly(nk))$ for each subroutine and the union bound ensures that all the applications to quantum subroutines success with high probability. This cause only a $\polylog(nk)$ factor and it is absorbed by the $\tilde{O}$ notation. 
\end{proof}

\subsection{Coreset Construction}\label{sec:coreset_details}

We present a quantum algorithm for constructing a coreset based on a bicriteria approximate solution. The construction is a quantum implementation of~\citet{cohen2021new}. \alg{coreset_details} shows a sketch of the construction.

\alg{coreset_details} requires an access to an ($\alpha$, $\beta$)-bicriteria approximation $A$, which means an oracle $O_A\colon \ket{i}\ket{0}\rightarrow \ket{i}\ket{a_i}$ $\forall i \in [m]$ for a set $A = \{a_1,\ldots,a_m\} \subset \mathbb{R}^d$, $m \leq \alpha k$, $\cost(D,A) \leq \beta \opt$. Based on $A$, using \lem{qANN}, one can construct an oracle
$$
O_{\tau}\colon \ket{s}\ket{0}\rightarrow\ket{s}\ket{i}, \quad \forall s \in [n]
$$
where $\dist(x_s,a_i) \leq c_{\tau}\dist(x_s,A)$ $\forall s \in [n]$ for some constant $c_{\tau}$. This oracle encodes a mapping $\tau\colon [n] \rightarrow [m]$ which maps each $x_s \in D$ to an approximately nearest neighbor $a_i \in A$ as its center. The map $\tau$ together with $A$ can be seen as a special solution for clustering. Let $\cost_{\tau}(D',A):= \sum_{x \in D'}\dist^z(x,a_{\tau(x)})$ be the cost of this solution and let $C_i := \{x\in D\mid \tau(x) = i\}$ as the $i$-th cluster induced by $\tau$ and $A$ for any $i \in [m]$.

\begin{lemma}\label{lem:coreset_details}
Let
$$
t = \tilde{O}\left( 2^{O(z)} \cdot m\cdot (d+\log(n))\cdot \max{(\varepsilon^{-2},\varepsilon^{-z})}\right)
$$
in \alg{coreset_details}. For a positive real $\varepsilon < 1/(4c_{\tau}^z)$, \alg{coreset_details} outputs an $O(c_{\tau}^{z}\beta \varepsilon)$-coreset of size $\tilde{O}\left( 2^{O(z)}md\log(n)\max{(\varepsilon^{-2},\varepsilon^{-z})}\right)$ with probability at least $5/6$, using $\tilde{O}\left(2^{O(z)}c_{\tau}\sqrt{nmd}\max{(\varepsilon^{-1},\varepsilon^{-z/2})}\right)$ queries to $O_{\tau}$, $O_D$, $O_A$, their inverses, and QRAM. Besides it uses $\poly(md\log n/\varepsilon^z)$ additional classical processing time.
\end{lemma}

\begin{algorithm}[tb]
\caption{Coreset Construction}\label{alg:coreset_details}
\begin{algorithmic}[1]
\STATE {\bfseries input:} $t$, $\varepsilon$, oracle $O_D$, $O_A$, $O_{\tau}$
\STATE {\bfseries output:} $O(c_{\tau}^{z}\beta\varepsilon)$-coreset $\Omega$ with weight function $w$ \\
// \quad  phase 1: partitioning the dataset into groups \\
\STATE compute $\Delta_{C_i} = \cost_{\tau}(C_i,A)/|C_i|$ for each $i \in [m]$ by using $O_{\tau}$ and \thm{mqsum}, store in QRAM
\STATE construct the ring unitary $U_R$ \\
\STATE compute $\cost_{\tau}(R_{i,j},A)$ and $\cost_{\tau}(R_j,A)$ for each pair $i,j$ by using $U_R$ and \thm{mqsum}, store in QRAM
\STATE construct the group unitary $U_G$ \\
//\quad  phase 2: sensitivity sampling and reweighting
\STATE compute $\cost_{\tau}(G_{j,b},A)$ for each pair of $j$ and $b$ using $U_{G}$ and \thm{mqsum}, store in QRAM
\STATE for each $i \in [m]$, compute $|R_{i,I}|+|C_i \cap (\cup_{j\neq I} G_{j,\min})|$ by \thm{mqcounting}, and assign the value to $w(a_i)$
\STATE for each well-structured group $G$, draw a size-$t$ i.i.d.~sample $\Omega$ such that each $x\in C_i\cap G$ is selected in each round with the same probability
$$
\Pr[x] = \frac{\cost_{\tau}(C_i,A)}{|C_i|\cost_{\tau}(G,A)}
$$
using \lem{qsampling} and reweight $w(x) = 1/(t\Pr[x])$ for each sampled point $x$
\STATE for each outer group $G$, draw a size-$t$ i.i.d.~sample $\Omega$ such that each each $x\in G$ is selected in each round with probability
$$
\Pr[x] = \frac{\cost_{\tau}(x,A)}{\cost_{\tau}(G,A)}
$$
using \lem{qsampling} and reweight $w(x) = 1/(t\Pr[x])$ for each sampled point $x$
\STATE let $\Omega$ be the union of all the above samples and $A$
\end{algorithmic}
\end{algorithm}

The details of \alg{coreset_details} are described as follows. This algorithm consists of two phases. During the first phase the algorithm partitions the dataset $D$ into groups. This consists of two steps. The first step of the first phase is to partition each cluster into rings, with each ring containing the points with the same distance from the center up to factor $2$. For each $C_i$, let 
$$
R_{i,j} := \{x \in C_i\mid 2^j\Delta_{C_i} \leq \cost_{\tau}(x,A) \leq 2^{j+1}\Delta_{C_i} \}.
$$
Let $R_{i,I} := \cup_{j \leq -2z\log(z/\varepsilon)}R_{i,j}$ be the inner ring and $R_{i,O} := \cup_{j > 2z \log (z/\varepsilon)}$ be the outer ring. Besides, let $R_I := \cup_{i = 1}^m R_{i,I}$, $R_O := \cup_{i = 1}^m R_{i,O}$, and $$R_j := \cup_{i = 1}^m R_{i,j}\quad\forall j,\  -2z\log(z/\varepsilon) < j \leq 2z \log (z/\varepsilon).$$
The ring unitary $U_R$ is defined as
$$
U_R\colon \ket{s}\ket{0}\ket{0}\rightarrow \ket{s}\ket{i}\ket{j} \quad \forall s \in [n]
$$
where $x_s \in R_{i,j}$ for each $s \in [n]$. For $j \leq -2z\log(z/\varepsilon)$, $U_R$ uses a same special notation for such $j$ and in this paper it is denoted as $j = I$. The special notation can be any preselected value out of $[-2z\log(z/\varepsilon), 2z \log (z/\varepsilon)]$. And it is the same for $j = O$. $U_R$ is a unitary which answers the corresponding ring $R_{i,j}$ for each query $x_s \in D$.

The second step is to gather the rings into groups such that the rings with equal cost up to factor $2$ are gathered together and prepared to be handled together in the second phase. For each $j \in \{\lceil -2z\log(z/\varepsilon)\rceil, \ldots, \lfloor 2z\log(z/\varepsilon) \rfloor +1\}\cup\{I,O\}$, let
$$
G_{j,b} := \cup_{i \in I_{j,b}} R_{i,j},
$$
where $I_{j,b}$ is the largest set such that for any $i \in I_{j,b}$,
$$
\cost_{\tau}(R_{i,j},A) \in (\frac{\varepsilon}{4z})^z\cdot\frac{\cost_{\tau}(R_j,A)}{m}\cdot [2^b,2^{b+1}].
$$
And let $G_{j,\min}:= \cup_{b \leq 0}G_{j,b}$ be the union of the cheapest groups, and $G_{j,\max}:= \cup_{b \geq z\log (4z/\varepsilon)}G_{j,b}$ be the union of the most expensive ones. The same notation as $j = I$ and $j = O$ is used for $b = \min$ and $b = \max$. The group unitary $U_G$ is defined as
$$
U_{G}\colon \ket{s}\ket{0}\ket{0}\rightarrow \ket{s}\ket{j}\ket{b} \quad \forall s \in [n],
$$
where $x_s \in G_{j,b}$ for any $s \in [n]$. $U_G$ answers the corresponding group $G_{j,b}$ for each query $x_s \in D$.

In the second phase the dataset seen as the union of three different kinds of points and these three parts are handle separately. The first kind contains the union of inner rings $R_I$ and the cheapest groups $G_{j,\min}$ $\forall j$; The second kind is all the well-structured groups $G_{j,b}$ with $-2z\log(z/\varepsilon) < j \leq 2z \log (z/\varepsilon)$  and $b = 1,\ldots,\max$; And the third kind is all the outer groups $G_{O,b}$ with $b = 1,\ldots,\max$. 

In quantum computing, it is costly to compute the exact sum such as $\cost_{\tau}(C_i,A)$ and $\cost_{\tau}(R_{i,j},A)$.
Hence, \alg{coreset_details} uses $\varepsilon$-estimations instead of corresponding exact values, but for convenience they are simply written as they are exact. It turns out that $O(\varepsilon)$-estimations are enough for constructing an $O(\varepsilon)$-coreset.

\begin{proof}[Proof of \lem{coreset_details}]

The proof is shown in \append{proof_coreset_details}.
\end{proof}

\begin{remark}
We note that in Line 3 (and similarly, Line 5 and 7), computing $\Delta_{C_i}$ for all the $m$ clusters in $\tilde{O}(\sqrt{nm})$ time  is feasible in quantum computing. We can compute $\{\cost_{\tau}(C_i,A)\}_{i = 1}^m$ and $\{|C_i|\}_{i = 1}^m$ separately, and then calculate the division in a classical method, so we only state for the calculation of all the $\cost_{\tau}(C_i,A)$. We can construct the following unitary $U$ by one query to $O_{\tau}$ and its inverse.
$$
U\colon \ket{s}\ket{0}\ket{0} \rightarrow \ket{s}\ket{\tau(s)}\ket{\dist^z(x_s,A)} \quad \forall s \in [n]
$$
Note that $\cost_{\tau}(C_i,A) = \sum_{\tau(s) = i}\dist^z(x_s,A)$. According to \thm{mqsum}, calculating these values requests only $\tilde{O}(\sqrt{nm})$ time. More details about \thm{mqsum} is shown in \sec{mqcounting}.
\end{remark}

Let $m = O(k\log n)$, $\beta = 2^{O(z)}$, $\varepsilon = \varepsilon'/2^{O(z)}$, and $c_{\tau} = 5/2$ in \lem{coreset_details}. Note that $O_A$ is obtained by storing $A$ in QRAM, and one query to $O_{\tau}$ uses $d\polylog(mn)$ queries to QRAM by \lem{qANN}. Combining \lem{bicriteria} and \lem{coreset_details} directly yields \thm{coreset}.

\section{Multidimensional Approximate Summation}\label{sec:mqcounting}

A crucial subroutine of \alg{coreset_details} is to compute the summation of the $\cost_{\tau}(x,A)$ over all the points $x$ in each part for a given partition (Line 3, 5, and 7). This gives rise to such a problem:
\begin{definition}[Multidimensional Approximate Summation]\label{def:masum}
Given two integers $1 \leq m \leq n$, a real parameter $\varepsilon > 0$, a partition $\tau\colon [n] \rightarrow [m]$, and a function $f\colon [n] \rightarrow \mathbb{R}_{\geq 0}$. The multidimensional approximate summation problem is to find $\varepsilon$-estimation for each $s_j := \sum_{\tau(i) = j}f(i)$, $j \in [m]$.
\end{definition}

This paper proposes \emph{multidimensional quantum approximate summation} to solve this problem. We believe this technique can have wide applications in designing quantum algorithms for machine learning and other relevant problems.
\begin{theorem}[Multidimensional Quantum Approximate Summation]\label{thm:mqsum}
Assume that there exists access to an oracle $O_{\tau}\colon \ket{i}\ket{0}\ket{0}\rightarrow\ket{i}\ket{\tau(i)}\ket{f(i)}$ $\forall i \in [n]$ and assume that $f$ has an upper bound $M$. For $\varepsilon \in (0,1/3)$, $\delta > 0$, there exists a quantum algorithm that solves the multidimensional approximate summation problem with probability at least $1-\delta$, using $\tilde{O}\left(\sqrt{nm/\varepsilon}\log(1/\delta)\log M\right)$ queries to $O_{\tau}$, $\tilde{O}\left((\sqrt{nm/\varepsilon}+m/\varepsilon)\log(n/\delta)\log M \right)$ gate complexity, and additional $O(m\log M)$ classical processing time.
\end{theorem}

$f(i)$ is a  binary number of length $\lceil\log M\rceil$. By first computing the summation of each digit and then summing up the results together, the multidimensional approximate summation problem can be reduced to the following problem:
\begin{definition}[Multidimensional Counting]\label{def:mqcounting}
Given two integers $1 \leq m \leq n$, a real parameter $\varepsilon > 0$, and a partition $\tau\colon [n] \rightarrow [m]$. For each $j \in [m]$, denote $D_j := \{i \in [n]: \tau(i)=j\}$ as the $j$-th part and $n_j := |D_j|$ for the size. The multidimensional counting problem is to find $\varepsilon$-estimation $\tilde{n}_j$ for each $n_j$, $j \in [m]$.
\end{definition}

For this problem, this section proposes \emph{multidimensional quantum counting} to solve it. 
\begin{theorem}[Multidimensional Quantum Counting]\label{thm:mqcounting}
Assume that there exists access to an oracle $O_{\tau}\colon \ket{i}\ket{0}\rightarrow\ket{i}\ket{\tau(i)}$ $\forall i \in [n]$. For $\varepsilon \in (0,1/3)$, $\delta > 0$, \alg{mqcounting} solves the multidimensional counting problem with probability at least $1-\delta$, using $\tilde{O}\left(\sqrt{nm/\varepsilon}\log(1/\delta)\right)$ queries to $O_{\tau}$ and additional $\tilde{O}\left((\sqrt{nm/\varepsilon}+m/\varepsilon)\log(n/\delta)\right)$ gate complexity. The query complexity is optimal up to a logarithm factor.
\end{theorem}

Denote $p_j := n_j/n$. Note that by one call for $O_{\tau}$ one can construct the unitary
\begin{equation}\label{eqn:oracle_p}
O_p\colon \ket{0}\rightarrow \sum_{j = 1}^m \sqrt{p_j}|j\rangle \left(\frac{1}{n_j}\sum_{i \in D_j}\ket{i}\right).
\end{equation}
Consider $p = (p_1,\ldots,p_m)$ as an $m$-dimensional probability distribution, the above unitary can be seen as a probability oracle to $p$. The following method is used to estimate the distribution:

\begin{lemma}[Multidimensional Amplitude Estimation, Rephrased from Theorem 5 and Lemma 7 of~\citealt{van2021quantum}]\label{lem:amplitudeestimation}
There is a quantum algorithm which has the following properties: given precision $\varepsilon \in (0,1/3)$, error probability $\delta > 0$, a quantum probability oracle $O_p$ on $q$ qubits for an $m$-dimensional probability distribution $p$, a number set $S \subset [m]$, and a constant $p_{mt} \geq \sum_{i \in S}p_i$ as the maximal total probability on $S$, it outputs $\tilde{p}\in\mathbb{R}^{m}$ such that
$$
|\tilde{p}_i - p_i| \leq \varepsilon\quad\forall i \in S
$$
with probability $\geq 1-\delta$ using $O\left({\varepsilon}\log(m/\delta)\sqrt{p_{mt}\varepsilon}\right)$ calls to $O_p$ and membership queries to $S$. The gate complexity is $\tilde{O}\left((q (m+p_{mt}/\varepsilon) +\sqrt{p_{mt}}/\varepsilon)\log(1/\delta)\right)$ using a QRAM.
\end{lemma}

A naive method is to set $p_{mt} = 1$ and apply \lem{amplitudeestimation} directly. However, it requests an $\varepsilon$ very small to ensure the size estimation of the small parts sufficiently precise, which leads to gratuitous overprecise estimation for the large parts and results in the waste of time. \alg{mqcounting} performs a trick to obtain the estimations hierarchically: in each iteration it sets a certain precision (Line 6), and saves only the estimation values large enough to ensure the accuracy and leaves the small parts to be estimated more precisely next time (Line 8-10).

\begin{algorithm}[tb]
\caption{Multidimensional Quantum Counting}\label{alg:mqcounting}
\begin{algorithmic}[1]
\STATE{\bfseries input:} $n$, $m$, $O_{\tau}$, $\varepsilon \in (0,\frac{1}{3})$, $\delta>0$
\STATE{\bfseries output:} $\{\tilde{n}_1,\cdots,\tilde{n}_m\}$, s.t. $|\hat{n}_j-n_j|\leq \varepsilon n_j$ $\forall j \in [m]$
\STATE initialize $P \leftarrow [m]$, $Q \leftarrow [n]$, $\tilde{n} \leftarrow n$, $p_{mt} \leftarrow 1$
\STATE construct the oracle $O_p$ in \eqn{oracle_p}
\REPEAT
 \STATE for each $j \in P$, estimate $\{p_j\}$ as $\{\tilde{p}_j\}$ by applying \lem{amplitudeestimation} with $q = \lceil \log n\rceil$, precision $\frac{2\tilde{n}\varepsilon}{3nm}$, maximal total probability $p_{mt}$, error probability $O(\frac{\delta}{\log n})$
 \FOR{$j \in P$}
 \IF{$\tilde{p_j}\geq \frac{\tilde{n}}{9nm}$}
 \STATE $\tilde{n}_j\leftarrow n\tilde{p}_j$, $P \leftarrow P-\{j\}$
 \ENDIF
 \ENDFOR
 \STATE $Q \leftarrow \{i\in[n]:\tau(i)\in P\}$\;
 \STATE make an $1/2$-estimation $\tilde{n}$ for $|Q|$ using \lem{counting}
 \STATE $p_{mt} \leftarrow \frac{2\tilde{n}}{n}$
\UNTIL{$\tilde{n} < m/\varepsilon$}
\STATE find all the items in $Q$ using \lem{getall}, count the remaining parts classically.
\end{algorithmic}
\end{algorithm}

The proof of \thm{mqcounting} and \thm{mqsum} is deferred to \append{proof_mqcounting}.

\begin{remark}
Note that our \alg{mqcounting} for multidimensional quantum computing is optimal up to a logarithmic factor due to \thm{lower-mqc-counting-main}.
\end{remark}

\section{Lower Bound}
To complement our quantum algorithms, we also prove the following quantum lower bounds. First, we have:

\begin{theorem}[Quantum Lower Bound for  Multidimensional Counting]\label{thm:lower-mqc-counting-main}
    Every quantum algorithm that solves the multidimensional counting problem (Definition \ref{def:mqcounting}) w.p. at least $\frac 2 3$ uses at least $\Omega\left(\sqrt{nk}\varepsilon^{-1/2}\right)$ queries to $O_\tau$.
\end{theorem}

The proof of \thm{lower-mqc-counting-main} is deferred to \append{proof-lower-mqc-counting-main}.
For the clustering problems, we establish the following lower bounds under different settings (proofs deferred to \append{proof-lower-clustering-main}).

\begin{theorem}[Quantum Lower Bounds for $k$-means and k-median]\label{thm:lower-clustering-main}
Assume that $\varepsilon$ is sufficiently small. Consider the Euclidean $k$-means/median problem on data set $D = \{x_1,\ldots,x_n\} \subset \mathbb R^d$. Assume a quantum oracle $O_x \ket{i, b} := \ket{i, b \oplus x_i}$. Then, every quantum algorithm outputs the followings with probability $2/3$ must have quantum query complexity lower bounds for the following problems:
\begin{itemize}[leftmargin=*]
    \item An $\varepsilon$-coreset: $\Omega\left(\sqrt{nk}\varepsilon^{-1/2}\right)$ for $k$-means and $k$-median (\thm{k-means-coreset-lower});
    
    \item An $\varepsilon$-estimation to the value of the objective function: $\Omega\left(\sqrt{nk}+\sqrt n \varepsilon^{-1/2}\right)$ for $k$-means and $k$-median (\thm{k-means-cost-lower});
    
    \item A center set $C$ such that $\cost(C) \le (1 + \varepsilon)\cost\left(C^*\right)$ where $C^*$ is the optimal solution: $\Omega\left(\sqrt{nk}\varepsilon^{-1/6}\right)$ for $k$-means; $\Omega\left(\sqrt{nk}\varepsilon^{-1/3}\right)$ for $k$-median (\thm{k-means-lower}). 
\end{itemize}
\end{theorem}

\section*{Acknowledgements}
This paper is partially supported by a national key R\&D program of China No. 2021YFA1000900 and a startup fund from Peking University.


\newcommand{\arxiv}[1]{arXiv:\href{https://arxiv.org/abs/#1}{\ttfamily{#1}}\?}\newcommand{\arXiv}[1]{arXiv:\href{https://arxiv.org/abs/#1}{\ttfamily{#1}}\?}\def\?#1{\if.#1{}\else#1\fi}

\bibliographystyle{icml2023}

\newpage
\appendix
\onecolumn

\section{Further Proof Details for Bicriteria Approximation}\label{append:proof_biapprox}

This section gives a rigorous proof of the correctness of the bicriteria approximate algorithm (\lem{bicriteria}), that the output of \alg{bicriteria} (the set $A$) is an  ($O(\log^2 n)$, $2^{O(z)}$)-bicriteria approximate solution with probability at least $5/6$. This proof follows~\citet{thorup2001quick}.

First, the loop (Line 4-10) stops for $t < 3\lceil \log n\rceil$ with high probability. In each iteration, assume all the points $x \in D_t$ are sorted $\dist(x,\tau_{t}(x))$ from small to large. The sample $s_t$ is drawn from $D_t$ uniformly at random (Line 6) and all the points $x$ preceding $s_t$ are deleted from $D_t$ (Line 7), so with probability at least $1/2$ it holds that $|D_{t+1}| \leq |D_t|/2$. If there's still $\tilde{r}_t > 39k\lceil \log n \rceil$ after $3\lceil \log n \rceil$ iterations, it holds that $|D_t| \geq 2\tilde{r}_t/3 > 26k\lceil \log n \rceil \geq 1$. Hence the event $|D_{t+1}| \leq |D_t|/2$ happens in no more than $\lceil \log n \rceil$ iterations, which is only $2/3$ of the expectation. The probability that such an event happens is at most $\exp\left(-(1.5 \log n)(1/3)^2/2\right) = \exp{\left(-(\log n)/12\right)} < n^{-0.12}$.

In the following, it is supposed that the loop (Line 4-10) stops for $t < 3\lceil \log n\rceil$. In this situation, it holds that
$$
|A| \leq 13k\lceil \log n\rceil \cdot 3\lceil \log n\rceil + 2\cdot39 \lceil \log n\rceil = O(k\log^2 n).
$$
For every $x \in D_t$, let $c_x$ be the corresponding center in the optimal solution. $x$ is called ``happy point" and this is denoted as $\Happy(x)$ if there exists $c \in A_{t+1}$ such that $\dist(c,c_x) \leq \dist(x,c_x)$. Otherwise $x$ is called ``angry point" and this event is denoted as $\Angry(x)$. For a certain $x \in D$, suppose that there are $q$ points in $D_t$ corresponding to $c_x$ in the optimal solution and as close to $c_x$ as $x$. Since in $A_{t+1}$ there are $13k\lceil \log n\rceil$ points sampled from $D_t$ uniformly at random, the probability that $x$ remains to be angry is no more that the probability that all the $q$ points haven't be selected, that is, 
$$
\Pr[\Angry(x)] \leq (1-\frac{q}{|D_t|})^{13k\lceil \log n\rceil} \leq \exp{\left(-\frac{q}{|D_t|}13k\lceil \log n\rceil\right)}.
$$
The expectation of the number of the angry points in $D_t$ corresponding to $c_x$ is at most 
$$
\sum_{q = 1}^{+\infty}\exp{\left(-\frac{q}{|D_t|}13k\lceil \log n\rceil\right)} \leq \int_{x = 0}^{+\infty}\exp{\left(-\frac{x}{|D_t|}13k\lceil \log n\rceil\right)}dx \leq \frac{|D_t|}{13k\lceil \log n\rceil}.
$$
Since there are $k$ centers in the optimal solution, the fraction of angry points in $|D_t|$ is $1/(13\lceil \log n\rceil)$.

For any happy point $x \in D_t$, there exists $c \in A_{t+1}$ such that $\dist(c,c_x) \leq \dist(x,c_x)$. Hence, it holds that
$$
\dist(x,A) \leq \dist(x,A_{t+1}) \leq \dist(x,c) \leq \dist(x,c_x) + \dist(c_x,c) \leq 2\dist(x,c_x).
$$

For the angry points, sort all the points $x \in D_t$ from small to large by the key $\dist(x,\tau_{t+1}(x))$. Let each angry point grabs the first ungrabbed happy point and assume that $s_t$ is an ungrabbed happy point. For any unhappy point $x \in D_t$, there is a happy point $y \in D_t$ preceding $s_t$ in the sequence grabbed by $x$, since otherwise $s_t$ is unhappy or grabbed and there comes a contradiction. Let $G_{t+1}$ be
$$
G_{t+1} := \{x \in D_t: \dist(x,\tau_{t+1}(x)) \leq \dist(s_t,\tau_{t+1}(s_t))\}.
$$
It can be concluded that for any unhappy point $x \in G_{t+1}$, there exists a unique happy point $y \in G_{t+1}$ such that $\dist(x,\tau_{t+1}(x)) \leq \dist(y,\tau_{t+1}(y))$. Hence, 
$$
\sum_{x \in G_{t+1}, \Angry(x)} \dist^z(x,\tau_{t+1}(x)) \leq \sum_{x \in G_{t+1}, \Happy(x) }\dist^z(x,\tau_{t+1}(x)).
$$
For all the points in $G_{t+1}$, it holds that
\begin{align*}
\cost(G_{t+1},A) = \sum_{x\in G_{t+1}}\dist^z(x,A) & \leq \sum_{x \in G_{t+1}}\dist^z(x,A_{t+1}) \\
& = \sum_{x \in G_{t+1}, \Angry(x) }\dist^z(x,A_{t+1}) + \sum_{x \in G_{t+1}, \Happy(x) }\dist^z(x,A_{t+1}) \\
& \leq \sum_{x \in G_{t+1}, \Angry(x)}\dist^z(x,\tau_{t+1}(x)) + \sum_{x \in G_{t+1}, \Happy(x) }\dist^z(x,A_{t+1})\\
& \leq  \sum_{x \in G_{t+1}, \Happy(x) }\dist^z(x,\tau_{t+1}(x)) + \sum_{x \in G_{t+1}, \Happy(x) }\dist^z(x,A_{t+1})\\
& \leq 2c_{\tau}^z \sum_{x \in G_{t+1}, \Happy(x) }\dist^z(x,A_{t+1}(x))\\
& \leq O(2^zc_{\tau}^z) \sum_{x \in G_{t+1}, \Happy(x) }\dist^z(x,c_x) \\
& \leq O(2^zc_{\tau}^z) \sum_{x \in G_{t+1}}\dist^z(x,c_x).
\end{align*}
The probability that $s_t$ is ungrabbed happy point is the fraction of ungrabbed happy points in $D_t$, which is $2/(13\lceil\log n\rceil)$. The probability that all the $s_t$ in the $3\lceil\log n\rceil$ iterations are ungrabbed and happy is at least $1-\frac{2}{13\lceil\log n\rceil}\cdot 3\lceil\log n\rceil = 7/13$. 

Due to the definition of $G_t$ and $D_t$, $\cup_{t}G_t$ contains all the points deleted during the loop (Line 4-10). Besides, Since all the remaining points are added into $A$, they do not contribute to the cost. With a probability $7/13 - n^{-0.12} > 1/2$, it holds that
\begin{align*}
\cost(D,A) = & \sum_{x \in D}\dist^z(x,A) \\
& = \sum_{t}\sum_{G_t} \dist^z(x,A) + \sum_{x \notin \cup_t G_t} \dist^z(x,A) \\
& \leq O(2^zc_{\tau}^z)\sum_{t}\sum_{G_t}d(x,c_x) + 0 \\
& \leq O(2^zc_{\tau}^z)\opt.
\end{align*}
In Line 12, \alg{bicriteria} repeats the processing in Line 3-11 for three times and union all the set $A$ to boost the success probability to $5/6$. As a consequence, with probability at least $5/6$, the output $A$ is an ($O(k\log^2n),O(2^zc_{\tau}^z)$)-bicriteria approximate solution. The proof of our complexity claim has been given in \sec{bicriteria}.

\section{Further Proof Details for Coreset Construction Based on Biapproximate Solution}\label{append:proof_coreset_details}

This section gives a detailed proof of \lem{coreset_details}. We restate this lemma as below.
\begin{lemma}
Let
$$
t = \tilde{O}\left( 2^{O(z)} \cdot m\cdot (d+\log(n))\cdot \max{(\varepsilon^{-2},\varepsilon^{-z})}\right)
$$
in \alg{coreset_details}. For a positive real $\varepsilon < 1/(4c_{\tau}^z)$, \alg{coreset_details} outputs an $O(c_{\tau}^{z}\beta \varepsilon)$-coreset of size $\tilde{O}\left( 2^{O(z)}md\log(n)\max{(\varepsilon^{-2},\varepsilon^{-z})}\right)$ with probability at least $5/6$, using $\tilde{O}\left(2^{O(z)}c_{\tau}\sqrt{nmd}\max{(\varepsilon^{-1},\varepsilon^{-z/2})}\right)$ queries to $O_{\tau}$, $O_D$, $O_A$, their inverses, and QRAM. Besides it uses $\poly(md\log n/\varepsilon^z)$ additional classical processing time.
\end{lemma}

This section first provides the detailed quantum implementation and the analysis of complexity in \append{implementation}, and then gives a proof that the output of \alg{coreset_details} (set $\Omega$) is an $O(c_{\tau}^z\beta\varepsilon)$-coreset with probability at least $5/6$ in \append{correctness}.

\subsection{Quantum Implementation and Complexity}\label{append:implementation}
This section shows the quantum implementation details with the complexity. 

\begin{proof}[Proof for complexity]
For Line 3, the algorithm estimates $|C_i|$ and $\cost_{\tau}(C_i,A)$ first. Constructing the oracle
\begin{align*}
U\colon & \ket{s}\ket{0}\ket{0}\ket{0}\ket{0}\xmapsto{O_{\tau}} \ket{s}\ket{i}\ket{0}\ket{0}\ket{0} \\
& \xmapsto{O_D,O_A}\ket{s}\ket{i}\ket{x_s}\ket{a_i}\ket{0} \\
& \mapsto \ket{s}\ket{i}\ket{x_s}\ket{a_i}\ket{\dist^z(x_s,a_i)} \\
& \xmapsto{O_A^{-1},O_D^{-1}} \ket{s}\ket{i}\ket{0}\ket{0}\ket{\dist^z(x_s,a_i)}
\end{align*}
and applying \thm{mqsum} yields the needed values $\cost_{\tau}(C_i,A)$ $\forall i \in [m]$. Since $\dist^z(x_s,a_i) \leq \cost_{\tau}(D,A) \leq c_{\tau}^z\opt$, the calculation uses no more than $\tilde{O}(z\log(c_{\tau})\sqrt{nm}/\varepsilon)$ queries to $U$ and additional time, under a fair assumption that $\opt = \poly(n)$. The same technique works for $|C_i|$. Then the algorithm computes $\Delta_{C_i}$ in a classical manner. The implementation of Line 5, Line 7, and Line 8 is similar. These calculation uses in total no more than $\tilde{O}(z\log(c_{\tau})\sqrt{nm}/\varepsilon)$ queries to $U_R$, $U_G$, $O_{\tau}$, $O_D$, and $O_A$. Besides it uses $\poly\left(mz\log(1/\varepsilon)\right)$ classical processing time.

The construction of the ring unitary $U_R$ in Line 4 is
\begin{align*}
U_R\colon & \ket{s}\ket{0}\ket{0}\ket{0}\ket{0}\ket{0}\\
& \xmapsto{O_{\tau}} \ket{s}\ket{i}\ket{0}\ket{0}\ket{0}\ket{0}\\
& \xmapsto{O_{\Delta},O_D,O_A} \ket{s}\ket{i}\ket{\Delta_{C_i}}\ket{x_s}\ket{a_i}\ket{0}\\
& \mapsto \ket{s}\ket{i}\ket{\Delta_{C_i}}\ket{x_s}\ket{a_i}\ket{j}\\
& \xmapsto{O_{A}^{-1},O_D^{-1},O_{\Delta}^{-1}} \ket{s}\ket{i}\ket{0}\ket{0}\ket{0}\ket{j}\\
\end{align*}
where $j = \lfloor \log \left(\dist^z(x_s,a_i)/\Delta_{C_i}\right)\rfloor$, and $O_{\Delta}\colon \ket{i}\ket{0}\rightarrow\ket{i}\ket{\Delta_{C_i}}$ $\forall i \in [m]$ is constructed by storing $\Delta_{C_i}$ in QRAM in Line 3. One query to $U_R$ needs constant queries to $O_D$, $O_A$,  $O_{\tau}$, and QRAM. The same technique works for the construction of $U_G$ and the complexity is also the same up to a constant factor.

For Line 9, the algorithm first construct the below unitary $U$ for each well-sturctured $G$.
\begin{align*}
U \colon & \ket{s}\ket{0}\ket{0}\ket{0}\ket{0}\ket{0}\ket{0} \\ 
& \xmapsto{U_{G}, O_{\tau}}\;\xmapsto{O_{\Delta}}\;\mapsto \ket{s}\ket{j}\ket{b}\ket{i}\ket{\Delta_{C_i}}\ket{I(x_s \in G)}\ket{p_s} \\
& \mapsto\;\xmapsto{O_{\Delta}}\;\xmapsto{O_{\tau}^{-1},U_G^{-1}} \ket{s}\ket{0}\ket{0}\ket{0}\ket{0}\ket{0}\ket{p_s}
\end{align*}
where $I(x_s \in G)$ is the indicator for whether $x_s \in G$ and $p_i = I(x_s\in G)\Delta_{C_i}$ is proportional to $\Pr[x_s]$. Then one application to \lem{qsampling} yields the sample $\Omega$ using $O(\sqrt{nt})$ queries to the above unitary $U$. Reweighting can be completed in a classical manner since $U$ is of small size. For Line 10 the algorithm uses the same technique. The sampling process in total needs $\tilde{O}(\sqrt{nt})$ queries to $U_G$, $O_A$, $O_{\tau}$, and QRAM. It uses additional $\poly(m\log n) + O(t)\cdot d\polylog(mn)$ for computing the weight classically.

Let $t = \tilde{O}\left( 2^{O(z)} \cdot m\cdot (d+\log(n))\cdot \max{(\varepsilon^{-2},\varepsilon^{-z})}\right)$ and sum up all the time cost. \alg{coreset_details} uses $\tilde{O}\left(2^{O(z)}c_{\tau}\sqrt{nmd}\max{(\varepsilon^{-1},\varepsilon^{-z/2})}\right)$ queries to $O_{\tau}$, $O_D$, $O_A$ and QRAM, and $\poly(md\log n/\varepsilon^z)$ additional classical processing time.

Similar to \alg{bicriteria}, each subroutine used in \alg{coreset_details} suffers only a $\log(1/\delta)$ factor to reach success probability at least $1-\delta$ and each subroutine is applied no more than $\poly(nmz)$ times, so it is enough to set the failure probability as $\delta = O(1/\poly(nmz))$ for each subroutine. This cause only a $\polylog(nmz)$ factor on time consume and it is adsorbed by the $\tilde{O}$ notation.
\end{proof}

\subsection{Correctness}\label{append:correctness}
This section gives a rigorous proof that set $\Omega$, the output of \alg{coreset_details}, is an $O(c_{\tau}^z\beta\varepsilon)$-coreset with probability at least $5/6$. This proof follows the idea of~\citet{cohen2021new}.

The dataset $D$ can be seen as a partition of the following three kinds of points:
\begin{itemize}
    \item the union of the inner rings $R_I$, the cheapest groups $G_{j,\min}$ with $j \in [z\log(4\varepsilon/z),z\log(4z/\varepsilon)]$ and $G_{O,\min}$
    \item the well-structured groups $G_{j,b}$ with $j \in [z\log(\varepsilon/4z),z\log(4z/\varepsilon)]$ and $b = 1,\ldots,\max$
    \item the outer rings $G_{O,b}$ with $b = 1,\ldots,\max$.
\end{itemize}
\alg{coreset_details} deals with this three kinds of points separately and so does the following proof. The following proof shows that, let $\varepsilon \leq 1/(4c_{\tau}^z)$ and
$$
t = \tilde{O}\left( 2^{O(z)} \cdot m\cdot (d+\log(n))\cdot \max{(\varepsilon^{-2},\varepsilon^{-z})}\right)
$$
in \alg{coreset_details}, this algorithm has the following three properties, which are stated formally and proved later.
\begin{itemize}
\item \textbf{\lem{first_kind}} Let $B := R_I\cup G_{j,\min}$ be the set of the first kind of points. For any $S \in (\mathbb{R}^d)^m$ it holds that
$$
|\cost(B,S) - \cost(A,S)| \leq 8\varepsilon\left(\cost(D,S) + \cost_{\tau}(D,A)\right).
$$
\item \textbf{\lem{well_group}} It holds with probability at least $1/12$ that for any well-structured group $G = G_{j,b}$ and the corresponding sample $\Omega = \Omega_{j,b}$, and for any $S \in (\mathbb{R}^d)^m$,
$$
|\cost(G,S) - \cost(\Omega,S)| = O(c_{\tau}^z\varepsilon)\left(\cost(G,S) + \cost(G,A)\right).
$$
\item \textbf{\lem{outer_rings}} It holds with probability at least $1/12$ that, for any outer group $G = G_{O,b}$ and the corresponding sample $\Omega_{O,b}$, and for any $S \in (\mathbb{R}^d)^m$,
$$
|\cost(G,S) - \cost(\Omega,S)| \leq \frac{2c_{\tau}^z\varepsilon}{z\log(z/\varepsilon)} (\cost(D,S) + \cost(D,A)).
$$
\end{itemize}

Note that there are only $O(z\log(z/\varepsilon))$ outer groups $G_{O,b}$. Combining the three properties directly yields the proof for correctness.
\begin{proof}[Proof for correctness]
$$
|\cost(D,S) - \cost(\Omega,S)| \leq O(c_{\tau}^z\varepsilon)\left(\cost(D,S) + \cost(D,A)\right) \leq O(c_{\tau}^z\beta\varepsilon)\cost(D,S).
$$
Therefore, the output of \alg{coreset_details} $\Omega = A\cup\Omega_{j,b}\cup\Omega_{O,b}$ is an $O(c_{\tau}^z\beta\varepsilon)$-coreset. 
\end{proof}

The two lemmas introduced as follows are important tools for the proof.
\begin{lemma}[Triangle Inequality of Powers]\label{lem:triangle}
Let $a$, $b$, and $c$ be three arbitrary sets of points $\mathbb{R}^d$. For any $z \in \mathbb{Z}_{+}$ and any $\varepsilon > 0$, it holds that
\begin{align*}
\dist^z(a,b) &\leq (1+\varepsilon)^{z-1} \dist^z(a,c) + \left(\frac{1+\varepsilon}{\varepsilon}\right)^{z-1}\dist^z(b,c) \\
|\dist^z(a,S) - \dist^z(b,S)| &\leq \varepsilon d^z(a,S) + \left(\frac{2z+\varepsilon}{\varepsilon}\right)^{z-1}\dist^z(a,b).
\end{align*}
\end{lemma}
\begin{lemma}[Rephrased from Definition 1 and Lemma 17 of~\citealt{cohen2021new}]\label{lem:centroid}
There is a set $\mathbb{C}$ of size $n\cdot (z/\varepsilon)^{O(d)}$ such that, for any solution $S \in (\mathbb{R}^d)^m$ there exists $\tilde{S} \in \mathbb{C}^m$ which ensures that for any point $x \in D$ with either $\cost(x,S) \leq (8z/\varepsilon)^z\cost_{\tau}(x,A)$ or $\cost(x,\tilde{S}) \leq (8z/\varepsilon)^z\cost_{\tau}(x,A)$, it holds that
$$
|\cost(x,S) - \cost(x,\tilde{S})| \leq \varepsilon\left(\cost(x,S)+\cost_{\tau}(x,A) \right).
$$
Such a set $\mathbb{C}$ is called an \emph{A-approximate centroid set} for ($m$, $z$)-clustering on data set $D$.
\end{lemma}
\begin{proof}[Proof of \lem{centroid}]
For any $x \in \mathbb{R}^d$ and $r \geq 0$, let $B(x,r) := \{y \in \mathbb{R}^d\mid \dist(x,y) \leq r\}$ be the \emph{ball} around $x$ with radius $r$. Note that Euclidean space $\mathbb{R}^d$ has \emph{doubling dimension} $O(d)$, which means in this metric space any ball of radius $2r$ can be covered by $2^{O(d)}$ balls of radius $r$. For any $V \subset \mathbb{R}^d$, a \emph{$\gamma$-net} of $V$ is a set of points $X \subset V$ such that for any $v\in V$ there exists $x \in X$such that $\dist(x,v) \leq \gamma$ and for any $x,y \in X$ it holds that $\dist(x,y) > \gamma$. In $\mathbb{R}^d$, a point set $V \subset \mathbb{R}^d$ with diameter $D$ has a $\gamma$-net with size $2^{O(d \log (D/\gamma))}$ \citep{gupta2003bounded}.

Let data points $x_1,\ldots,x_n$ be in order with non-descreasing value of $\dist(x,a_{\tau(x)})$. Let $N_i$ be an $\varepsilon\cdot \dist(x,a_{\tau(x)})/(4z)$-net of $B\left(x_i, 10z\dist(x_i,a_{\tau(x_i)})/\varepsilon\right) \setminus \cup_{j < i} B\left(x_j, 10z\dist(x_j,a_{\tau(x_j)})/\varepsilon\right)$. Let $s_f \in \mathbb{R}^d$ be a point such that $s_f \notin B\left(x_i, 10z\dist(x_i,a_{\tau(x_i)})/\varepsilon\right)$ $\forall i \in [n]$. Let $N := \cup_{x_i \in D}N_i \cup \{s_f\}$.

The size of $N$ is bound by $n\cdot (z/\varepsilon)^{O(d)}$:
$$
|N| \leq n \cdot 2^{O(d\log (z/\varepsilon)^2) } = n\cdot \left(\frac{z}{\varepsilon}\right)^{O(d)}.
$$

$N$ is an $A$-approximated centroid set. For any solution $S \in (\mathbb{R}^d)^m$, let $\tilde{S} \in N^m$ be constructed by the following method. For each point $s \in S$, let $i$ be the smallest index such that $s \in B\left(x_i, 10z\dist(x_i,a_{\tau(x_i)})/\varepsilon\right)$. The corresponding $N_i$ is non-empty because otherwise there exists $x_j$ such that $j < i$ and $s \in B(x_j,10z\dist(x_j,a_{\tau(x_j)})/\varepsilon)$, and thus $i$ is not the smallest index. Let $\tilde{s}$ be the closest point to $s$ in $N_i$. If such an index $i$ does not exist, let $\tilde{s} = s_f$. Let $\tilde{S}$ be the set of all the $\tilde{s}$. $\tilde{S}$ has the property defined in \lem{centroid}.

Let $x \in D$ satisfies $\cost(x,S) \leq (10z/\varepsilon)^z\cost_{\tau}(x,A)$. Let $s$ be the nearest neighbor of $x$ in $S$ and consider the corresponding index $i$ and $\tilde{s}$. It holds that $\dist(x,a_{\tau(x)}) \geq \dist(x_i,a_{\tau(x_i)})$ since $s \in B\left(x, 10z\dist(x,a_{\tau(x)})/\varepsilon\right)$. By the definition of $\tilde{s}$ it holds that $\dist(s,\tilde{s}) \leq \varepsilon\dist(x_i,a_{\tau(x_i)})/(4z) \leq (\varepsilon/4z)\dist_{\tau}(x,A)$. As a consequence,
\begin{align*}
\cost(x,\tilde{S}) & \leq \cost(x,\tilde{s}) \leq (1+\varepsilon)\cost(p,s) + (1+z/\varepsilon)^{z-1}\cost(s,\tilde{s}) \\
& \leq (1+\varepsilon)\cost(x,S) + \varepsilon \cost_{\tau}(x,A).
\end{align*}

On the other hand, let $x \in D$ satisfies $\cost(x,\tilde{S}) \leq (10z/\varepsilon)^z\cost_{\tau}(x,A)$. Let $\tilde{s}$ be the nearest neighbor of $x$ in $\tilde{S}$ and consider the corresponding $s$ and index $i$. If the index of $x$ is smaller than $i$ it can be implied that $\tilde{s} \notin N_i$ because of the definition of $N_i$ and $\tilde{s} \in B\left(x, 10z\dist(x,a_{\tau(x)})/\varepsilon\right)$. As a consequence, $\dist(x,a_{\tau(x)}) \geq \dist(x_i,a_{\tau(x_i)})$. It holds that
\begin{align*}
\cost(x,S) & \leq \cost(x,s) \leq (1+\varepsilon)\cost(x,\tilde{s}) + (1+2z/\varepsilon)^{z-1}\cost(s,\tilde{s}) \\
& \leq (1+\varepsilon)\cost(x,\tilde{S}) + \varepsilon \cost_{\tau}(x,A).
\end{align*}

For any $x \in D$ with $\cost(x,S) \leq (8z/\varepsilon)^z \cost_{\tau}(x,A)$, it holds that
$$
\cost(x,\tilde{S}) \leq (1+\varepsilon)\cost(x,S) + \varepsilon\cost_{\tau}(x,A) \leq (10z/\varepsilon)^z\cost_{\tau}(x,A).
$$
Thus $\cost(x,S) \leq (1+\varepsilon)\cost(x,\tilde{S}) + \varepsilon\cost_{\tau}(x,A)$. Hence
$$
|\cost(x,S) - \cost(x,\tilde{S})| \leq \varepsilon \left(\cost(x,S) + \cost_{\tau}(x,A) \right).
$$
The same inequality holds for any $x \in D$ with $\cost(x.\tilde{S}) \leq (8z/\varepsilon)^z\cost_{\tau}(x,A)$.
\end{proof}


\paragraph{The first kind of points} Let $B := R_I\cup G_{j,\min} \cup G_{\min}^O$ be the set of the first kind of points. There holds the following lemma:
\begin{lemma}\label{lem:first_kind}
For any solution $S$, $|S| \leq m$ and $\varepsilon < 1/2$ it holds that
$$
|\cost(B,S) - \cost(A,S)| \leq 8\varepsilon\left(\cost(D,S) + \cost_{\tau}(D,A)\right).
$$
\end{lemma}

We use the following two lemmas to prove \lem{first_kind}:

\begin{lemma}\label{lem:inner_ring}
For any solution $S$, $|S| \leq m$, any $i \in [m]$, and $\varepsilon < 1/2$,
$$
|\cost(R_I(C_i),S) - |R_I(C_i)|\cdot\cost(a_i,S)| \leq \varepsilon\cost(R_I(C_i),S) +2 \varepsilon \cost_{\tau}(C_i,A).
$$
\end{lemma}
\begin{lemma}\label{lem:cheapest_well_group}
For any solution $S$, $|S| \leq m$ and any cheapest group $G$,
$$
|\cost(G,S) - \sum_{i = 1}^m |C_i \cap G| \cost(a_i,S) | \leq \varepsilon\cost(R_j,S) + \varepsilon\cost_{\tau}(R_j,A)
$$
if $G = G_{j,\min}$ for some $j \neq O$. And if $G = G_{\min}^O$ there is
$$
|\cost(G,S) - \sum_{i = 1}^m |C_i \cap G| \cost(a_i,S) | \leq \varepsilon\cost(D,S) + \varepsilon\cost_{\tau}(D,A).
$$
\end{lemma}

\begin{proof}[Proof of \lem{inner_ring}]
Fix $i$. Using \lem{triangle}, it holds that
$$
|\cost(x,S) - \cost(a_i,S)| = |\dist^z(x,S) - \dist^z(a_i,S)| \leq \varepsilon \dist^z(x,S) + (1+ \frac{2z}{\varepsilon})^{z-1}\dist^z(a_i,x).
$$
For any $x \in R_I(C_i)$ and $\varepsilon < 1/2$, there is
$$
(1+2z/\varepsilon)^{z-1}\dist^z(a_i,x) \leq  (1+2z/\varepsilon)^{z-1}\cost_{\tau}(x,A) \leq (\frac{3z}{\varepsilon})^{z-1}(\frac{\varepsilon}{4z})^{z} \Delta_{C_i} \leq \varepsilon \frac{\cost_{\tau}(C_i,A)}{(1-\varepsilon)|C_i|} \leq 2\varepsilon \frac{\cost_{\tau}(C_i,A)}{|R_I(C_i)|}
$$
Combining the two inequalities above and summing the result over all the points in $R_I$ yields
\begin{align*}
|\cost(R_I(C_i),S) - |R_I(C_i)|\cdot \cost(a_i,S)| & \leq \sum_{x\in R_I(C_i)} |\cost(x,S) - \cost(a_i,S)| \\
& \leq \varepsilon\cost(R_I(C_i),S) +2 \varepsilon \cost_{\tau}(C_i,A).
\end{align*}
\end{proof}

\begin{proof}[Proof of \lem{cheapest_well_group}]
Using \lem{triangle}, it holds that
$$
|\cost(x,S) - \cost(a_{\tau(x)},S)| = |\dist^z(x,S) - \dist^z(a_{\tau(x)},S)| \leq \varepsilon \dist^z(x,S) + \left(1+ \frac{2z}{\varepsilon}\right)^{z-1}\dist^z(a_{\tau(x)},x).
$$
The sum over $G$ tells
\begin{align*}
|\cost(G,S) - \sum_{i = 1}^m |C_i \cap G| \cost(a_i,S) | & = |\sum_{x \in G}\cost(x,S) - \sum_{x \in G}\cost(a_{\tau(x)},S)| \\
& \leq \sum_{x \in G}\left(\varepsilon \dist^z(x,S) + (1+ \frac{2z}{\varepsilon})^{z-1}\dist^z(a_{\tau(x)},x)\right) \\
& \leq \varepsilon\cost(G,S) + (\frac{3z}{\varepsilon})^{z-1}  \cost_{\tau}(G,A).
\end{align*}
If $G = G_{j,\min}$ for some $j$, it holds that $\cost_{\tau}(G,A) \leq (\varepsilon/4z)^z\cdot \cost_{\tau}(R_j,A)$. Else $G = G_{\min}^O$ and $\cost_{\tau}(G,A) \leq (\varepsilon/4z)^z\cdot \cost_{\tau}(R_O,A) \leq (\varepsilon/4z)^z\cdot\cost_{\tau}(D,A)$. In both cases the lemma holds.
\end{proof}

Combining \lem{inner_ring} and \lem{cheapest_well_group} straightforwardly gives \lem{first_kind}. As is talked about in \sec{coreset_details}, due to the particularity of quantum computing, it is costly to compute the exact value $|R_{i,I}|+|C_i \cap (\cup_{j\neq I} G_{j,\min})|$. What the algorithm uses is $\varepsilon$-estimations $\tilde{r_i}$ such that 
$$
|\tilde{r}_i - (|R_I(C_i)|+|C_i\cap (\cup_{j\notin\{I,O\}} G_{j,\min})|+|C_i\cap G_{\min}^O|) \leq \varepsilon (|R_I(C_i)|+|C_i\cap (\cup_{j\notin\{I,O\}} G_{j,\min})|+|C_i\cap G_{\min}^O|)
$$
for any $i \in [m]$.

\begin{proof}[Proof of \lem{first_kind}]

\begin{align*}
& |\cost(B,S) - \cost(A,S)|  \\
= & |\cost(B,S) - \sum_{i = 1}^m \tilde{r}_i\cost(a_i,S)| \\
=  & |(1\pm\varepsilon) \cost(B,S) \mp \varepsilon\cost(B,S) - (1\pm \varepsilon)\sum_{i = 1}^m (|R_I(C_i)| + |C_i \cap (\cup_j G_{j,\min} )| + |C_i \cap G_{\min}^O|)\cost(a_i,S) | \\
\leq & (1+\varepsilon)\left(\sum_{i = 1}^m|(\cost(R_I(C_i),S) - |R_I(C_i)|\cost(a_i,S)|)  + \sum_{G} (|\cost(C\cap G,S) -\sum_{i = 1}^m |C_i\cap G|\cost(a_i,S)| )     \right) \\
& +\varepsilon\cost(B,S) \\
\leq & (1+\varepsilon)\varepsilon\left(\sum_{i = 1}^m (\cost(R_I(C_i),S)+2\cost_{\tau}(C_i,A)) + \sum_{j}(\cost(R_j,S) + \cost_{\tau}(R_j,A))) + \cost(D,S) + \cost_{\tau}(D,A)   \right) \\
& + \varepsilon\cost(B,S) \\
\leq & (1+\varepsilon)\varepsilon(3\cost(D,S)+4\cost_{\tau}(D,A))+\cost(D,S) \\
\leq & 8\varepsilon\left(\cost(D,S) + \cost_{\tau}(D,A)\right).
\end{align*}

\end{proof}


\paragraph{Well-structured groups}
For well-structured groups, the following lemma holds:
\begin{lemma}\label{lem:well_group}
Let
$$
t = \tilde{O}\left( 2^{O(z)} \cdot m\cdot (d+\log(n))\cdot \max{(\varepsilon^{-2},\varepsilon^{-z})}\right)
$$
in \alg{coreset_details} Line 9. For $\varepsilon < 1/(4c_{\tau}^z)$, it holds with probability at least $1/12$ that for any well-structured group $G = G_{j,b}$ and the corresponding sample $\Omega = \Omega_{j,b}$, and for any $S \in (\mathbb{R}^d)^m$,
$$
|\cost(G,S) - \cost(\Omega,S)| = O(c_{\tau}^z\varepsilon)\left(\cost(G,S) + \cost(G,A)\right).
$$
\end{lemma}

Fix a well-structured group $G$ and for convenience, in the following for any set $C \subset \mathbb{R}^d$ we write $G\cap C$ simply as $C$. Due to the definition of $G$, for every cluster $C_i$, the following properties hold:
\begin{itemize}
    \item $\forall x,y \in C_i$, $\cost_{\tau}(x,A) \leq 2\cost_{\tau}(y,A)$,
    \item $\cost_{\tau}(G,A)/(2m) \leq \cost_{\tau}(C_i,A)$,
    \item $\forall x \in C_i$, $\cost_{\tau}(C_i,A)/(2|C_i|) \leq \cost_{\tau}(x,A) \leq (2\cost_{\tau}(C_i,A))/|C_i|$.
\end{itemize}

Recall that $\Omega$ is an i.i.d sample of size $t$ and in each round a point $x \in C_i \cap G$ is sampled with probability
$$
\Pr[x] = \frac{\cost_{\tau}(C_i,A)}{|C_i|\cost_{\tau}(G,A)}
$$
and for any $x\in \Omega$,
$$
w(x) = \frac{|C_i|\cost_{\tau}(G,A)}{|\Omega|\cost_{\tau}(C_i,A)}.
$$
To be precisely, the values such as $|C_i|$ and $\cost_{\tau}(C_i,A)$ are $\varepsilon$-estimations in practice, instead of the exact values shown above. The method to deal with such problems is the same as in the proof of \lem{first_kind}. For convenience we do not repeat similar proof and use the exact values directly.

It can be seen that given a sample large enough, $|C_i|$ can be well approximated for every $i \in [m]$ (\lem{event_estimation}). The proof of \lem{event_estimation} is given later.
\begin{lemma}\label{lem:event_estimation}
Define event $\mathcal{E}$ to be for any $i \in [m]$,
$$
\sum_{x \in C_i\cap\Omega}\frac{|C_i|\cost_{\tau}(G,A)}{|\Omega|\cost_{\tau}(C_i,A)} = (1\pm\varepsilon)|C_i|.
$$
With probability at least $1-2m\cdot\exp{(-(\varepsilon^2 t)/(6m))}$, event $\mathcal{E}$ holds.
\end{lemma}

For any solution $S$, let
$$
I_{l,S} := \{x \in G\mid 2^l \cost_{\tau}(x,A) \leq \cost(x,S) \leq 2^{l+1}\cost_{\tau}(x,A) \},
$$
and let these $\{I_{l,S}\}$ be divided into three parts: tiny ranges, with $l \leq \log (\varepsilon/2)$; interesting ranges, with $\log(\varepsilon/2) \leq l \leq z\log(4z/\varepsilon)$; and huge ranges, with $l\geq z\log(4z/\varepsilon)$. For the three types of $I_{l,S}$, there exist the following lemmas, respectively:
\begin{lemma}\label{lem:tiny}
Let $I_{\tin,S} := \cup_{l \leq \log(\varepsilon/2)}I_{l,S}$. For any solution $S$, it holds that
$$
\max{\left(\cost( I_{\tin,S},S),\cost( I_{\tin,S} \cap \Omega,S )       \right) } \leq \varepsilon\cost_{\tau}(G,A).
$$
\end{lemma}
\begin{lemma}\label{lem:huge}
Condition on event $\mathcal{E}$, it holds that
$$
|\cost(C_i,S) - \cost(C_i\cap \Omega,S)|\leq O(\varepsilon)\cost(C_i,S)
$$
for any solution $S$, and any cluster $C_i$ such that there exists a huge range $I_{l,S}$ with $I_{l,S}\cap C_i \neq \varnothing$.
\end{lemma}
\begin{lemma}\label{lem:interesting}
Let $L_S := \{C_i\mid \forall x \in C_i, \cost(x,S) \leq (4z/\varepsilon)^z\cost(x,A)\}$. Let $\mathbb{C}$ be an $A$-approximate centroid set of size $n\cdot(z/\varepsilon)^{O(d)}$ as defined in \lem{centroid}. With probability
$$
1 - \exp{\left(m\log(n) + O(md\log(z/\varepsilon)) - 2^{O(z\log z)}\cdot \frac{\min(\varepsilon^2,\varepsilon^z)}{\log^2 (1/\varepsilon)}\cdot t  \right)}
$$
and together with event $\mathcal{E}$, it holds that for any solution $S \in \mathbb{C}^m$
$$
|\cost(L_S,S) - \cost(L_S\cap \Omega,S)| \leq \varepsilon(\cost_{\tau}(G,A) + \cost(G,S)).
$$
\end{lemma}

Using \lem{huge} and \lem{interesting}, \lem{well_group} can be proved.
\begin{proof}[Proof of \lem{well_group}]
Let $G$ be an arbitrary well-structured group. Let $S$ be an arbitrary solution and let $\tilde{S} \in \mathbb{C}^k$ approximate $S$. Denote 
\begin{align*}
& H_S := \{x \in G\mid \exists i, x \in C_i \mbox{ and } \exists l > z\log (8z/\varepsilon), C_i \cap I_{l,S} \neq \varnothing\} \\ & H_{\tilde{S}} : = \{x \in G\mid \exists i, x \in C_i \mbox{ and } \exists l > z\log (4z/\varepsilon), C_i \cap I_{l,\tilde{S}} \neq \varnothing\} \setminus H_S.
\end{align*}
And denote $L_{\tilde{S}}$ as in \lem{interesting}. It holds that $H_S$, $H_{\tilde{S}}$, and $L_{\tilde{S}}$ form a partition of $G$. On the one hand, $H_S \cup H_{\tilde{S}} \cup L_{\tilde{S}} = G$. On the other hand, $H_S \cap H_{tilde{S}} = \varnothing$, $H_{\tilde{S}} \cap L_{\tilde{S}} = \varnothing$, and $L_{\tilde{S}}\cap H_S = \varnothing$ since $\forall x \in L_{\tilde{S}}$ $\cost(x,\tilde{S}) \leq (4z/\varepsilon)^z \cost_{\tau}(x,A)$, thus $\cost(x,S) \leq (1+\varepsilon)\cost(x,\tilde{S}) + \varepsilon\cost_{\tau}(x,A) \leq (8z/\varepsilon)^z \cost_{\tau}(x,A)$.

By this partition and the property of $\tilde{S}$, it holds that
\begin{align*}
& |\cost(G,S) - \cost(\Omega,S)|\\
= & |\sum_{x \in H_S}\cost(x,S) - \sum_{x \in H_S \cap \Omega}w(x)\cost(x,S) | + |\sum_{x \in G\setminus H_S}\cost(x,S) - \sum_{x \in (G\setminus H_S)\cap \Omega}w(x)\cost(x,S)| \\
\leq & |\sum_{x \in H_S}\cost(x,S) - \sum_{x \in H_S \cap \Omega}w(x)\cost(x,S) | + |\sum_{x \in G\setminus H_S}\cost(x,\tilde{S}) - \sum_{x \in (G\setminus H_S)\cap \Omega}w(x)\cost(x,\tilde{S})| \\
& \quad + \varepsilon \left(\cost(G,S) + \cost_{\tau}(G,A) + \cost(\Omega,S) + \cost_{\tau}(\Omega,A)\right) \\
\leq & |\sum_{x \in H_S}\cost(x,S) - \sum_{x \in H_S \cap \Omega}w(x)\cost(x,S) | + \varepsilon \left(\cost(G,S) + \cost_{\tau}(G,A) + \cost(\Omega,S) + \cost_{\tau}(\Omega,A)\right)  \\
& \quad + |\sum_{x \in L_{\tilde{S}}}\cost(x,S) - \sum_{x \in L_{\tilde{S}}\cap \Omega}w(x)\cost(x,S)| + |\sum_{x \in H_{\tilde{S}}}\cost(x,S) - \sum_{x \in H_{\tilde{S}}\cap \Omega}w(x)\cost(x,S)| \\
\leq & O(\varepsilon)\left(\cost(G,S) + \cost_{\tau}(G,A) + \cost(\Omega,S) + \cost_{\tau}(\Omega,A) \right) \\
\leq & O(c_{\tau}^z\varepsilon)\left(\cost(G,S) + \cost(G,A) + \cost(\Omega,S) + \cost(\Omega,A) \right).
\end{align*}
The last inequality uses \lem{huge} and \lem{interesting}.

Assume that $\varepsilon < 1/(4c_{\tau}^z)$. Let $S = A$, it holds that
$$
\cost(\Omega,A) \leq \cost(G,A) + |\cost(G,A) - \cost(\Omega,A)| \leq O(1) \cost(G,A).
$$
Similarly,
$$
\cost(\Omega,S) \leq \cost(G,S) + |\cost(G,S) - \cost(\Omega,S)| \leq O(1)\left(\cost(G,S) + \cost(G,A) \right).
$$
Hence it can be concluded that
$$
|\cost(G,S) - \cost(\Omega,S)| = O(c_{\tau}^z\varepsilon)\left(\cost(G,S) + \cost(G,A)\right)
$$
Using the union bound over event $\mathcal{E}$ and the probability of \lem{interesting} for all the well-structured groups $G$, the probability is
$$
1 - z^2\log^2(z/\varepsilon)\left(\exp{\left(m\log(n) + O(md\log(z/\varepsilon)) - 2^{O(-z\log z)}\cdot \frac{\min(\varepsilon^2,\varepsilon^z)}{\log^2 (1/\varepsilon)}\cdot t  \right)} - 2m\cdot\exp{(-(\varepsilon^2 t)/(6m))}\right).
$$
The probability can be bound by $1/12$ by setting the value of $t$ as
$$
t = \tilde{O}\left( 2^{O(z)} \cdot m\cdot (d+\log(n))\cdot \max{(\varepsilon^{-2},\varepsilon^{-z})}\right).
$$
\end{proof}

The proofs of\lem{event_estimation}, \lem{tiny}, \lem{huge}, and \lem{interesting} are shown as below, respectively. \lem{event_estimation} is used in the proof of \lem{huge} and \lem{interesting}, and \lem{tiny} is used in the proof of \lem{interesting}.

\begin{proof}[Proof of \lem{event_estimation}]
Fix $i \in [m]$. Define $P_i(x)$ as the indicator of point $x \in \Omega$ being drawn from $C_i$, i.e., $P_i(x) = 1$ if $x \in C_i\cap\Omega$, and otherwise $P_i = 0$. The expectation of $P_i(x)$ has the following property:
$$
\E(P_i(x)) = \sum_{x \in C_i} |\Omega|\Pr[x] = \sum_{x \in C_i}\frac{|\Omega|\cost_{\tau}(C_i,A)}{|C_i|\cost_{\tau}(G,A)} \geq \frac{|\Omega|}{2m}.
$$
By Chernoff bounds, it holds that
$$
\Pr[|\sum_{x \in \Omega}P_i(x) - \E(P_i(x))| \geq \varepsilon \E(P_i(x)) ] \leq 2e^{-\varepsilon^2\E(P_i(x))/3} \leq 2e^{-(\varepsilon^2|\Omega|)/(6m)}.
$$
The union bound over all the clusters derives that, with probability at least $1-2m\cdot\exp{(-(\varepsilon^2 t)/(6m))}$, for any cluster $C_i$ there is
$$
|C_i\cap \Omega| = (1\pm\varepsilon)\sum_{x \in C_i}\frac{|\Omega|\cost_{\tau}(C_i,A)}{|C_i|\cost_{\tau}(G,A)}
$$
which implies
$$
\sum_{x \in C_i\cap\Omega}\frac{|C_i|\cost_{\tau}(G,A)}{|\Omega|\cost_{\tau}(C_i,A)} = (1\pm\varepsilon)|C_i|.
$$
\end{proof}

\begin{proof}[Proof of \lem{tiny}]
\begin{align*}
\cost(I_{\tin,S},S) & = \sum_{x \in I_{\tin,S}}\cost(x,S)
 \leq |I_{\tin,S}|\frac{\varepsilon}{2}\cost_{\tau}(x,A) \leq \varepsilon\cost_{\tau}(G,A) \\
 \cost(I_{\tin,S}\cap \Omega,S) & = \sum_{x \in I_{\tin,S}\cap \Omega}w(x)\cost(x,S) \\
 & = \sum_{x \in I_{\tin,S}\cap\Omega}\frac{|C_i|\cost_{\tau}(G,A)}{t\cost_{\tau}(C_i,A)}\cost(x,S) \\
 & \leq |I_{\tiny,S}\cap \Omega| \frac{|C_i|\cost_{\tau}(G,A)}{|\Omega|\cost_{\tau}(C_i,A)}\cdot\frac{\varepsilon}{2}\cdot\frac{2\cost_{\tau}(C_i,A)}{|C_i|} \\
 & \leq \varepsilon\cost_{\tau}(G,A).
\end{align*}
\end{proof}

\begin{proof}[Proof of \lem{huge}]
For any $x \in C_i\cap \Omega$, $\cost(x,S)$ is bound. Let $y \in I_{l,S}\cap C_i$ with $I_{l,S}$ being a huge range. For any $x \in C_i$, there is
\begin{align*}
\cost(x,y) \leq \left(\dist(x,A)+\dist(y,A) \right)^z \leq 3^z\cost_{\tau}(y,A) \leq 3^z2^{l-z\log(4z/\varepsilon)}\cost_{\tau}(y,A)\leq \left(\frac{3\varepsilon}{4z}\right)^z\cost(y,S) .
\end{align*}
Using \lem{triangle}, there is
\begin{align*}
\cost(y,S) & \leq \left(1+\frac{\varepsilon}{2z}\right)^{z-1}\cost(x,S) + \left(1+\frac{2z}{\varepsilon}\right)^{z-1}\cost(x,y) \\
& \leq (1+\varepsilon)\cost(x,S) + \varepsilon\cost(y,S)
\end{align*}
which implies $\cost(x,S) \geq (1-2\varepsilon)\cost(y,S)$ and $\cost(y,S) \leq (1+3\varepsilon)\cost(x,S)$ if $\varepsilon < 1/3$. Similarly $\cost(x,S) \leq (1+2\varepsilon)\cost(y,S)$ and $\cost(y,S) \geq (1-3\varepsilon)\cost(x,S)$.
\begin{align*}
\cost(C_i\cap\Omega,S) & = \sum_{x \in C_i\cap\Omega}\frac{|C_i|\cost_{\tau}(G,A)}{|\Omega|\cost_{\tau}(C_i,A)}\cost(x,S) \\
& = (1\pm2\varepsilon)\sum_{x \in C_i\cap\Omega}\frac{|C_i|\cost_{\tau}(G,A)}{|\Omega|\cost_{\tau}(C_i,A)}\cost(y,S) \\
& = (1\pm2\varepsilon)(1\pm\varepsilon)|C_i|\cost(y,S) \\
& = (1\pm2\varepsilon)(1\pm\varepsilon)\sum_{x \in C_i}\cost(y,S) \\
& = (1\pm2\varepsilon)(1\pm\varepsilon)(1\pm3\varepsilon)\cost(C_i,S)\\
& = (1\pm O(\varepsilon))\cost(C_i,S).
\end{align*}
The third equation holds because of event $\mathcal{E}$.
\end{proof}

\begin{proof}[Proof of \lem{interesting}]
Let $x_{i,S} := \arg\min_{x \in C_i}\cost(x,S)$ and $w_{x,S} = \left(\cost(x,S)-\cost(x_{i,S},S)\right)/\cost_{\tau}(x_{i,S},A)$. Let $E_{l,S} := \sum_{C_i\in L_S}\sum_{x \in C_i\cap I_{l,S}\cap \Omega}w(x)\cost_{\tau}(x_{i,S},A)w_{x,S}$ and $F_{l,S} := \sum_{C_i\in L_S}\sum_{x \in C_i\cap I_{l,S}\cap\Omega}w(x)\cost(x_{i,S},S)$.

$w_{x,S}$ is bound. For fixed $i$, $l$ and $S$ consider arbitrary $x \in C_i\cap I_{l,S}$. By the definition of $x_{i,S}$ it is straightforward to see $\cost(x,S) \geq \cost(x_{i,S},S)$, and thus $w_{x,S} \geq 0$. Besides, because the property of well-structured group there are
$$
\cost(x_{i,S},S) \leq \cost(x,S) \leq 2^{l+1}\cost_{\tau}(x,A) \leq 2^{l+2}\cost_{\tau}(x_{i,S},A)
$$
and
$$
\cost(x,x_{i,S}) \leq 2^{z-1}(\cost(x,A)+\cost(x_{i,S},A)) \leq 3\cdot2^{z-1}\cost_{\tau}(x_{i,S},A).
$$
Using \lem{triangle}, for any $\alpha \leq 1$,
$$
\cost(x,S) \leq (1+\frac{\alpha}{z})^{z-1}\cost(x_{i,S},S) + (1+\frac{z}{\alpha})^{z-1}\cost(x,x_{i,S})
$$
which after rearranging implies
\begin{align*}
\cost(x,S) - \cost(x_{i,S},S) & \leq 2\alpha\cost(x_{i,S},S) + (\frac{2z}{\alpha})^{z-1}\cost(x,x_{i,S}) \\ 
& \leq 2^z(2\alpha\max(1,2^{l+1})+(\frac{2z}{\alpha})^{z-1})\cost_{\tau}(x_{i,S},A)
\end{align*}
Let $\alpha = 2^{-l/z}$ (ignoring constants that depend on $z$), the inequality yields that
$$
\cost(x,S) - \cost(x_{i,S},S) \leq 2^{O(z\log z)}2^{l(1-1/z)}\cost_{\tau}(x_{i,S},A).
$$
Therefore, $w_{x,S} \in [0, 2^{O(z\log z)}2^{l(1-1/z)}]$.

$E_{l,S}$ can be expressed differently:
\begin{align}
E_{l,S} & = \sum_{C_i \in L_S}\sum_{x \in C_i \cap I_{l,S} \cap \Omega} w(x)\cost_{\tau}(x_{i,S},A) w_{x,S} \\
& = \sum_{C_i \in L_S}\sum_{x \in C_i \cap I_{l,S} \cap \Omega} w(x)(\cost(x,S) - \cost(x_{i,S},S)) \\
& = \sum_{x \in I_{l,S}\cap L_S\cap \Omega} w(x)\cost(x,S) - F_{l,S}.
\end{align}
The expectation of $E_{l,S}$ is as follows:
\begin{align*}
\E[E_{l,S}] & =  \sum_{x \in I_{l,S}\cap L_S} |\Omega|\Pr[x]w(x)\cost(x,S) - \E[F_{l,S}] \\
& = \sum_{x \in I_{l,S}\cap L_S}\frac{|\Omega|\cost_{\tau}(C_i,A)}{|C_i|\cost_{\tau}(G,A)}\frac{|C_i|\cost_{\tau}(G,A)}{|\Omega|\cost_{\tau}(C_i,A)}\cost(x,S) - \E[F_{l,S}] \\
& = \cost(I_{l,S}\cap L_S,S) - \E[F_{l,S}].
\end{align*}
Intuitively, by Bernstein's inequality the random variable $E_{l,S}$ is concentrated around its expectation. Let $\Omega_i$ be the point sampled from the $i$-th round of importance sampling (Line 9 in \alg{coreset_details}), and let $X_i = w(\Omega_i)\cost_{\tau}(x_{\tau(\Omega_i),S},A)w_{\Omega_i,S}$ when $\Omega_i \in I_{l,S}\cap \Omega$ and $X_i = 0$ otherwise. It holds that $E_{l,S} = \sum_{i = 1}^t X_i$. $X_i$ and its variance are bounded.
\begin{align*}
\Var[X_i] \leq \E[X_i^2] & = \sum_{x \in I_{l,S}\cap L_S} \Pr[x]\left(w(x)\cost_{\tau}(x_{\tau(x),S},A)w_{x,S}\right)^2  \\
& \leq \sum_{x \in I_{l,S}\cap L_S} \frac{\cost_{\tau}(C_i,A)}{|C_i|\cost_{\tau}(G,A)}\left(\frac{|C_i|\cost_{\tau}(G,A)}{|\Omega|\cost_{\tau}(C_i,A)}\cost_{\tau}(x,A)2^{l(1-1/z)}2^{O(z\log z)}   \right)^2 \\
& \leq \sum_{x \in I_{l,S}\cap L_S} 2^{2l(1-1/z)}2^{O(z\log z)} \frac{|C_i|\cost_{\tau}(G,A)}{|\Omega|^2\cost_{\tau}(C_i,A)}\cost_{\tau}^2(x,A) \\
&\leq  \sum_{x \in I_{l,S}\cap L_S} 2^{2l(1-1/z)}2^{O(z\log z)} \frac{\cost_{\tau}(G,A)}{|\Omega|^2}\cost_{\tau}(x,A)
\end{align*}
which implies
$$
\Var[X_i] \leq 
\left \{\begin{aligned}
& \frac{2^{O(z\log z)}\cost_{\tau}(G,A)\cost_{\tau}(G,A)}{|\Omega|^2} , \quad z = 1 \\
& \frac{2^{O(z\log z)}2^{l(1-2/z)}\cost_{\tau}(G,A)\cost(I_{l,S},S)}{|\Omega|^2} , \quad z \geq 2
\end{aligned}
\right.
$$
Besides, $X_i$ has upper bound.
\begin{align*}
X_i & \leq \frac{|C_i|\cost_{\tau}(G,A)}{|\Omega|\cost_{\tau}(C_i,A)}\cost_{\tau}(x,A)2^{l(1-1/z)}2^{O(z\log z)} \\
& \leq 2^{l(1-2/z)}2^{O(z\log z)} \frac{\cost_{\tau}(G,A)}{|\Omega|}.
\end{align*}
By Bernstein's inequality,
$$
\Pr[|E_{l,S} - \E[E_{l,S}]| \leq \frac{\varepsilon}{z\log z/\varepsilon} \cdot\left(\cost_{\tau}(G,A)+\cost(I_{l,S},S)   \right)] \leq \exp{\left(-\frac{\min(\varepsilon^2,\varepsilon^z)t}{2^{O(z\log z)}\log^2(1/\varepsilon)}\right)}.
$$

Denote $F_S := \sum_{l \leq z\log(4z/\varepsilon)}F_{l,S}$. Condition on event $\mathcal{E}$, the value of $F_S$ and its expectation are as follows:
\begin{align*}
\E[F_S] & = \sum_{l \leq z \log(4z/\varepsilon)} \sum_{C_i \in L_S}\sum_{x \in C_i\cap I_{l,S}} |\Omega|\Pr[x]w(x)\cost(x_{i,S},S) \\
& = \sum_{C_i \in L_S} |C_i|\cost(x_{i,S},S); \\
F_S & = \sum_{C_i \in L_S}\sum_{x \in C_i\cap \Omega} w(x)\cost(x_{i,S},S) \\
& = \sum_{C_i \in L_S}\cost(x_{i,S},S)\sum_{x \in C_i\cap \Omega}\frac{|C_i|\cost_{\tau}(G,A)}{|\Omega|\cost_{\tau}(C_i,A)}\\
& = (1\pm \varepsilon)\sum_{C_i\in L_S}|C_i|\cost(x_{i,S},S).
\end{align*}
Hence $F_S = (1\pm\varepsilon)\E[F_S]$, and $\E[F_S] \leq \cost(L_S,S) \leq \cost(G,S)$.

Taking an union bound over the concentration for all possible $S \in \mathbb{C}^k$ and all $l$ such that $\log(\varepsilon/2) \leq l \leq z \log (4z/\varepsilon)$, it holds with probability $1-\exp{\left(k\log (|\mathbb{C}|) - 2^O(z\log z)\cdot  \min(\varepsilon^2,\varepsilon^z)\cdot t \cdot \log^{-2}(1/\varepsilon)  \right)}$ that, for every $S \in \mathbb{C}^k$ and $\log(\varepsilon/2) \leq l\leq z\log (4z/\varepsilon)$,
$$
|E_{l,S} - \E[E_{l,S}] \leq \frac{\varepsilon}{z\log (z/\varepsilon)}\left(\cost_{\tau}(G,A) + \cost(I_{l,S},S) \right)|.
$$
Conditioning on the above event together with event $\mathcal{E}$, It holds that
\begin{align*}
|\cost(L_S,S) - \cost(L_S\cap\Omega,S)| & = |\sum_{x \in L_S}\cost(x,S) - \sum_{x \in L_S\cap \Omega}w(x)\cost(x,S)|\\
& \leq |\sum_{x \in L_S}\cost(x,S) - \E[F_S] + F_S - \sum_{x \in L_S\cap\Omega}w(x)\cost(x,S)| + |\E[F_S] - F_S| \\
& \leq \sum_{l < \log (\varepsilon/2)} |\sum_{x \in I_{l,S}\cap L_S}\cost(x,S) - \E[F_{l,S}] + F_{l,S} - \sum_{x \in I_{l,S}\cap L_S\cap\Omega}w(x)\cost(x,S)| \\
& \quad + \sum_{l = \log(\varepsilon/2)}^{z\log (4z/\varepsilon)}|\sum_{x \in I_{l,S}\cap L_S}\cost(x,S) - \E[F_{l,S}] + F_{l,S} - \sum_{x \in I_{l,S}\cap L_S\cap\Omega}w(x)\cost(x,S)| \\
& \quad + |\E[F_S]-F_S|.
\end{align*}
The tiny ranges can be bound as follows since $F_{l,S} \leq \sum_{x \in I_{l,S}\cap \Omega} w(x)\cost(x,S)$ and $\E[F_{l,S}] \leq \sum_{x \in I_{l,S}}\cost(x,S)$:
\begin{align*}
& \sum_{l < \log (\varepsilon/2)} |\sum_{x \in I_{l,S}\cap L_S}\cost(x,S) - \E[F_{l,S}] + F_{l,S} - \sum_{x \in I_{l,S}\cap L_S\cap\Omega}w(x)\cost(x,S)| \\
\leq & \sum_{l < \log(\varepsilon/2)}\left(\sum_{x \in I_{l,S}\cap L_S}\cost(x,S) + \E[F_{l,S}] + F_{l,S} + \sum_{x \in I_{l,S}\cap L_S\cap\Omega}w(x)\cost(x,S)\right) \\
\leq & 2\left(\sum_{x \in I_{l,S}}\cost(x,S) + \sum_{I_{l,S}\cap \Omega}w(x)\cost(x,S) \right) \\
\leq & 4\varepsilon \cost_{\tau}(G,A).
\end{align*}
Plugging this result into the previous inequality, it holds that
\begin{align*}
& |\cost(L_S,S) - \cost(L_S\cap\Omega,S)| \\
\leq & 4\varepsilon\cost_{\tau}(G,A) + \sum_{l = \log (\varepsilon/2)}^{z\log (4z/\varepsilon)}|E_{l,S} - \E[E_{l,S}]| + |\E[F_S] - F_S| \\
 \leq & 4\varepsilon\cost_{\tau}(G,A) + \left(z\log (4z/\varepsilon)-\log(\varepsilon/2)\right)\cdot\frac{\varepsilon}{z\log(z/\varepsilon)}\left(\cost_{\tau}(G,A) + \cost(L_S,S)\right) + \varepsilon\cost(G,S) \\
 \leq & O(\varepsilon)(\cost_{\tau}(G,A) + \cost(G,S)).
\end{align*}
\end{proof}

\paragraph{Outer groups}
For outer groups, the following lemma holds:
\begin{lemma}\label{lem:outer_rings}
Let
$$
t = \tilde{O}\left(2^{O(z)}\cdot m\cdot (d+\log n)\cdot\frac{1}{\varepsilon^2}\right)
$$
in \alg{coreset_details} Line 10.
It holds with probability at least $1/12$ that, for any group of outer rings $G = G_b^O$ and the corresponding sample $\Omega_b^O$, and for any $S \in (\mathbb{R}^d)^m$,
$$
|\cost(G,S) - \cost(\Omega,S)| \leq 2\frac{c_{\tau}^z\varepsilon}{z\log(z/\varepsilon)} (\cost(D,S) + \cost(D,A)).
$$
\end{lemma}

\begin{proof}

Fix an arbitrary $S$. Partition the points in $G$ into two parts and denote
\begin{align*}
& G_{\close,S} := \{x \in G\mid \cost(x,S) \leq 4^z\cost_{\tau}(x,A)\} \\
& G_{\far,S} := \{x \in G\mid \cost(x,S) > 4^z\cost_{\tau}(x,A)\}.
\end{align*}

Bernstein's inequality works for the close part. Let $\Omega_i$ be the $i$-th sampled point and let
$$
X_i = \left \{\begin{aligned}
& \frac{\cost_{\tau}(G,A)}{|\Omega|\cost_{\tau}(\Omega_i,A)}\cdot\cost(\Omega_i,S), \quad \Omega_i \in G_{\close,S} \\
& 0, \quad \Omega_i \notin G_{\close,S}
\end{aligned}
\right.
$$
Let $E_{\close,S}:= \sum_{i = 1}^t X_i$. The variance of $X_i$ has the property
\begin{align*}
\Var[X_i] \leq \E[X_i^2] & = \sum_{x \in G_{\close,S}}\left(\frac{\cost_{\tau}(G,A)}{|\Omega|\cost_{\tau}(x,A)}\cost(x,S) \right)^2 \frac{\cost_{\tau}(x,A)}{\cost_{\tau}(G,A)} \\
& =\frac{\cost_{\tau}(G,A)}{|\Omega|^2} \sum_{x \in G_{\close,S}} \frac{\cost(x,S)}{\cost_{\tau}(x,A)}\cost(x,S) \\
& \leq \frac{4^z}{|\Omega|^2}\cost_{\tau}(G,S)\cost(G,S).
\end{align*}
$X_i$ has a upper bound
$$
X_i \leq \max_{x \in G_{\close,S}}( \frac{\cost_{\tau}(G,A)}{|\Omega|\cost_{\tau}(x,A)}\cost(x,S)) \leq \frac{4^z}{|\Omega|}\cost_{\tau}(G,A).
$$
Bernstein's inequality yields that
$$
\Pr[|E_{\close,S} - \E[E_{\close,S}]|\leq \frac{\varepsilon}{z\log (z/\varepsilon)}(\cost_{\tau}(D,A) + \cost(D,S))] \leq \exp\left( -2^{O(z)}\cdot (\frac{\varepsilon}{z\log(z/\varepsilon)})^2 \cdot t\right)
$$
Similar technique about $A$-approximate centroid set yields that for any $S$ and any $G$,
$$
|\cost(G_{\close,S},S) - \cost(G_{\close,S}\cap \Omega,S)| \leq \frac{c_{\tau}^z\varepsilon}{z\log(z/\varepsilon)}\left(\cost(D,A)+\cost(D,S)\right)
$$
with probability at least
$$
1 - \exp\left( m\log n + O(md)\log (\frac{z^2}{\varepsilon}\log(z/\varepsilon)) - 2^{-O(z)}\varepsilon^2t\right).  
$$

The proof for the far part is as follows. Denote event $\mathcal{E}_{\far}$ to be: for any cluster $C$,
$$
\sum_{x \in C\cap G\cap \Omega} w(x)\cost_{\tau}(x,A) = (1\pm\varepsilon)\cost_{\tau}(C\cap G,A).
$$
Event $\mathcal{E}_{far}$ happens with probability at least $1 - m\exp(\varepsilon^2t/m)$. Let $E_C = \sum_{i = 1}^t X_i$, where
$$
X_i = \left \{\begin{aligned}
& w(x)\cost_{\tau}(\Omega_i,A), \quad \Omega_i \in C\cap G \\
& 0, \quad \Omega_i \notin C\cap G
\end{aligned}
\right.
$$
with $\Omega_i$ being the $i$-th sampled point. Calculation shows that $\Var[X_i] \leq E[X_i^2] \leq 2m\cost^2_{\tau}(C\cap G,A)/t^2$ and $X_i \leq 2m\cost_{\tau}(C\cap G,A)/t$, and the Bernstein's inequality implies the success probability.

Fix a cluster $C_i$ such that $C_i \cap G_{\far,S} \neq \varnothing$ and let $a_i$ be the center. Let $x_i$ be a point such that $x_i \in C_i\cap G_{\far,S}$, which implies $\dist(x_i,S) \geq 4\dist(x_i,a_i)$. Let $C_{\close} : = \{x \in C_i\mid \cost_{\tau}(x,A)\leq (z/\varepsilon)^z(\cost_{\tau}(C_i,A)/|C_i|) \}$. Due to Markov's inequality, $|C_{\close}|\geq (1-\varepsilon/z)|C_i|$. Note that $G$ is an outer group and for any $x \in G$ it holds that $\cost_{\tau}(x,A) \geq (z/\varepsilon)^{2z}\cdot (\cost_{\tau}(C_i,A)/|C_i|)$. By the definition of $G_{\far,S}$ and \lem{triangle} it can be derived that
\begin{align*}
 \cost(a_i,S)& \geq (\dist(x_i,S) - \dist(x_i,a_i))^z \geq 3^z\cost_{\tau}(x_i,A) \geq 3^z(z/\varepsilon)^{2z}\frac{\cost_{\tau}(C_i,A)}{|C_i|} \\
\cost(a_i,S) & \leq (1+\varepsilon)\cost(x,S) + (1+2z/\varepsilon)^{z-1}\cost(x,a_i) \\
& \leq (1+\varepsilon)\cost(x,S) + (3z/\varepsilon)^{z-1}(z/\varepsilon)^z\frac{\cost_{\tau}(C_i,A)}{|C_i|} \\
& \leq (1+\varepsilon)\cost(x,S) + \varepsilon\cost_{\tau}(a_i,S)
\end{align*}
where $x$ is an arbitrary point in $C_{\close}$. Combining the lower bound of the size of $C_{\close}$, this implies
$$
\cost(C_i,S) \geq \cost(C_{\close,S}) |C_{\close}|\cdot\frac{1-\varepsilon}{1+\varepsilon}\cdot\cost(a_i,S)\geq 3^z(z/\varepsilon)^{2z-1}\cost_{\tau}(C_i,A).
$$
Due to Markov's inequality the size of $G\cap C_i$ is bound by $(\varepsilon/z)^2|C_i|$ since $G$ is an outer group. Combining with the above inequalities, it holds that
\begin{align*}
\cost(G_{\far,S}\cap C_i,S) & = \sum_{x \in G_{\far,S}\cap C_i}\cost(x,S) \\
& \leq \sum_{x \in G_{\far,S}\cap C_i}(1+\varepsilon)\cost(a_i,S) + (1+2z/\varepsilon)^{z-1}\cost_{\tau}(G_{\far,S}\cap C_i,A)\\
& \leq \left(\frac{1+\varepsilon}{1-\varepsilon}\right)^z\left(\frac{\varepsilon}{z}\right)^2\cost(C_i,S) + \left(\frac{3z}{\varepsilon}\right)^{z-1}\left(\frac{1}{3}\right)^z\left(\frac{\varepsilon}{z}\right)^{2z-1}\cost(G_{\far,S}\cap C_i,S).
\end{align*}
By simplifying the inequality and summing over all the clusters $C_i$, it is implied that
$$
\cost(G_{\far,S},S) \leq \frac{\varepsilon}{z\log(z/\varepsilon)}\cost(D,S).
$$
Condition on $\mathcal{E}_{\far}$, calculation shows that
$$
\cost(G_{\far,S}\cap \Omega, S) \leq \frac{\varepsilon}{z\log(z/\varepsilon)}\cost(D,S).
$$
The combination of the results about the close part and the far part tells that
\begin{align*}
& |\cost(G,S) - \cost(G\cap\Omega,S)| \\
\leq & |\cost(G_{\close,S},S)- \cost(G_{\close,S}\cap\Omega,S)|+|\cost(G_{\far,S},S)| + |\cost(G_{\far,S}\cap\Omega,S)| \\
\leq & \frac{2c_{\tau}^z\varepsilon}{z\log(z/\varepsilon)}(\cost(D,S)+\cost(D,A)).
\end{align*}
To make the above inequality holds with probability at least $1/12$, it is sufficient to set
$$
t = \tilde{O}\left(2^{O(z)}\cdot m\cdot (d+\log n)\cdot\frac{1}{\varepsilon^2}\right).
$$
\end{proof}

\section{Proof of Multidimensional Quantum Counting}\label{append:proof_mqcounting}
This section provides a detailed proof of \thm{mqcounting} and \thm{mqsum}.

\thm{mqcounting} is restated as follows:
\begin{theorem}\label{thm:mqcounting_appendix}
Given two integers $1 \leq m \leq n$, two parameters $\varepsilon \in (0,1/3)$, and a partition $\tau\colon [n] \rightarrow [m]$. For each $j \in [m]$, denote $D_j := \{i \in [n]\colon \tau(i)=j\}$ as the $j$-th part and $n_j := |D_j|$ for the size. Assume that we have an oracle $O_{\tau}\colon \ket{i}\ket{0}\rightarrow\ket{i}\ket{\tau(j)}$ $\forall i \in [n]$. For $\varepsilon \in (0,1/3)$, $\delta > 0$, \alg{mqcounting} outputs $\tilde{n}_j$ such that $|\tilde{n}_j-n_j| \leq \varepsilon n_j$ for every $j\in [m]$ with probability at least $1-\delta$, using $\tilde{O}\left(\sqrt{nm/\varepsilon}\log(1/\delta)\right)$ queries to $O_{\tau}$, $\tilde{O}\left((\sqrt{nm/\varepsilon}+m/\varepsilon)\log(n/\delta)\log M \right)$ gate complexity, and additional $O(m\log M)$ classical processing time. The query complexity is optimal up to a logarithm factor.
\end{theorem}

\begin{proof}
We establish the correctness and complexity bound separately.
\paragraph{Correctness}
The ``maximal total probability" $p_{mt}$ and precision $\frac{2\tilde{n}\varepsilon}{3nm}$ in Line 6 satisfies \lem{amplitudeestimation}. Denote the exact cardinality of $Q$ in Line 13 as $n'$, $|\tilde{n} - n'| \leq \frac{1}{2} n'$. For the maximal total probability, we have $\sum_{j \in S}p_j = \sum_{j \in S}\frac{n_j}{n} = \frac{n'}{n} \leq \frac{2\tilde{n}}{n} = p_{mt}.$ And for the precision, $\frac{2\tilde{n}\varepsilon}{3nm} \leq \frac{n'}{nm}\varepsilon \leq \varepsilon < 1/3$.

For each $j \in M$, $n_j$ has been estimated. For those $n_j$ estimated after the while loop stops, we find all the members belonging to the corresponding subset and count classically in Line 16, which gives an exact cardinality.

For those $n_j$ estimated in the while loop, $|\tilde{p}_j - \frac{n_j}{n}| = |\tilde{p_j} - p_j| \leq \frac{2\tilde{n}\varepsilon}{3nm}$. Since $\tilde{p_j}\geq \frac{\tilde{n}}{nm}$, $|\tilde{p}_j - \frac{n_j}{n}| \leq\frac{2}{3} \varepsilon \tilde{p_j}$, thus $|\tilde{n}_j - n_j| \leq \frac{2}{3} \varepsilon \tilde{n}_j$. Since $\varepsilon \in (0,1/3)$, $|\tilde{n}_j - n_j| \leq\varepsilon n_j$ as required.

\paragraph{Complexity}
In each iteration, the estimation process for $\{p_j\}_{j\in P}$ in Line 6 uses $O({\sqrt{p_{mt}}}/\frac{2\tilde{n}\varepsilon}{3nm}) = O(\frac{m\sqrt{n}}{\varepsilon\sqrt{\tilde{n}}}) =  O(\frac{\sqrt{nm}}{\varepsilon})$ applications of $U_p$ and membership queries for $P$, and $\tilde{O}(m+p_{mt}/\frac{2\tilde{n}\varepsilon}{3nm}) = \tilde{O}(\frac{m}{\varepsilon})$ gate complexity, since $\tilde{n} < m/\varepsilon$ (Line 15) at this time. The estimation for the cardinality of $Q$ needs $\tilde{O}(\sqrt{n/m}/\varepsilon)$ membership queries to $P$ according to \lem{counting}. The classical process in Line 7-11 and Line 16 needs at most $O(m)$ time. Since we can make all elements in $P$ sorted in $O(m\log(m))$ time to keep an $O(\log(m))$ query complexity for the membership query to $P$, the total query complexity per iteration is $\tilde{O}(\sqrt{\frac{nm}{\varepsilon}})$, with additional $\tilde{O}(\frac{m}{\varepsilon})$ processing time.

The loop (Line 5-15) has at most $O(\log n)$ iterations. Denote $\tilde{n}_0 = n$ and $\tilde{n}$ obtained in Line 13 at the $t$-th iteration as $\tilde{n}_t$. At the $(t+1)$-th iteration, for any $j \in P$ there is $\tilde{p_j} \leq \frac{\tilde{n}_t}{9nm}$, thus $p_j \leq \tilde{p}_j + \frac{2\tilde{n}_t\varepsilon}{3nm} \leq \frac{\tilde{n}_t}{3nm}$. We can calculate that $\tilde{n}_{t+1} \leq \frac{3}{2}n'_{t+1} = \frac{3n}{2}\sum_{j \in P}p_j \leq \frac{3nm}{2}\max_{j \in P}p_j \leq \frac{\tilde{n}_t}{2}$. As a result, after $O(\log n)$ iterations there must be $\tilde{n} \leq m/\varepsilon$.

Finding all the items remaining in $Q$ (Line 16) requires $\tilde{O}(\frac{\sqrt{nm}}{\varepsilon})$ queries to $O_{\tau}$ and membership queries to $P$ since $|P| \leq m/\varepsilon$ here. The overall query complexity is $\tilde{O}(\frac{\sqrt{nm}}{\varepsilon})$.

There are at most $O(\log n)$ iterations and each iteration fails with probability at most $O(\frac{\delta}{\log n})$. Therefore, the success probability is at least $1-\delta$.
\end{proof}

\thm{mqsum} is restated as follows:
\begin{theorem}[Multidimensional Quantum Approximate Summation]
Given two integers $1 \leq m \leq n$, a real parameter $\varepsilon > 0$, a partition $\tau\colon [n] \rightarrow [m]$, and a function $f\colon [n] \rightarrow \mathbb{R}_{\geq 0}$. Assume that there exists access to an oracle $O_{\tau}\colon \ket{i}\ket{0}\ket{0}\rightarrow\ket{i}\ket{\tau(i)}\ket{f(i)}$ $\forall i \in [n]$ and assume that $f$ has an upper bound $M$. For $\varepsilon \in (0,1/3)$, $\delta > 0$, there exists a quantum algorithm that finds $\varepsilon$-estimation for each $s_j := \sum_{\tau(i = j)}f(i)$, $j \in [m]$ with probability at least $1-\delta$, using $\tilde{O}\left(\sqrt{nm/\varepsilon}\log(1/\delta)\log M\right)$ queries to $O_{\tau}$ and additional $\tilde{O}\left((\sqrt{nm/\varepsilon}+m/\varepsilon)\log(n/\delta)\log M \right)$ gate complexity.
\end{theorem}

\begin{proof}
Write $f(i)$ as a binary number $\bar{f_0(i)f_1(i)\ldots f_l(i)}$, $l = \lceil \log M \rceil$. For each $t = 0:l$, let $O_t$ be
$$
O_t \colon \ket{i}\ket{0} \rightarrow \ket{i}\ket{\tau(i)I(f_t(i) = 1)},
$$
where $I(f_t(i) = 1)$ is the indicator for whether $f_t(i) = 1$. $O_t$ can be constructed by constant queries to $O_{\tau}$ and its inverse:
$$
\ket{i}\ket{0}\ket{0}\ket{0} \xmapsto{O_{\tau}} \ket{i}\ket{\tau(i)}\ket{f(i)}\ket{0} \mapsto \ket{i}\ket{\tau(i)}\ket{f(i)}\ket{\tau(i)I(f_t(i))} \xmapsto{O_{\tau}^{-1}} \ket{i}\ket{0}\ket{0}\ket{\tau(i)I(f_t(i) = 1)}.
$$
Applying \thm{mqcounting_appendix} with oracle $O_t$ and $\delta' = \delta/l$, \alg{mqcounting} outputs $\tilde{s}_j^t$ for $j = 1:m$, $\tilde{s}_j^t$ is an $\varepsilon$-estimation for $s_j^t = \sum_{\tau(i) =j} f_t(i)$ using $\tilde{O}\left(\sqrt{nm/\varepsilon}\log(l/\delta)\right)$ calls for $O_t$ and additional $\tilde{O}\left((\sqrt{nm/\varepsilon}+m/\varepsilon)\log(nl/\delta)\right)$ gate complexity. Let $\tilde{s}_j = \sum_{t = 0}^l 2^t \tilde{s}_j^t$.
$$
|\tilde{s}_j - s_j| \leq \sum_{t = 0}^l 2^t|\tilde{s}_j^t - s_j^t| \leq \varepsilon\sum_{t = 0}^l 2^t s_j^t \leq \varepsilon s_j.
$$
Therefore, $\tilde{s}_j$ is an $\varepsilon$-estimation of $s_j$, $\forall j \in [m]$. The total complexity is $\tilde{O}\left(\sqrt{nm/\varepsilon}\log(1/\delta)\log M\right)$ queries to $O_{\tau}$, $\tilde{O}\left((\sqrt{nm/\varepsilon}+m/\varepsilon)\log(n/\delta)\log M \right)$ gate complexity, and additional $O(m\log M)$ classical processing time.
\end{proof}

\section{Proofs of Quantum Lower Bounds}\label{append:lower}
\subsection{Auxiliary Lemmas}
In the proofs of our quantum lower bounds, we use the following tools.

\begin{theorem}[The Perfect Composition Theorem, \citealt{HoyerLS07}, \citealt{LeeMRSS11}, \citealt{Kimmel13}, and \citealt{Reichardt14}]\label{thm:composition_theorem}
    For the alphabet set $\Sigma,\Gamma$, functions $f: D_1 \to K$ and $g : \mathcal D_2 \to \Gamma$ with $\mathcal D_1 \subseteq \Gamma^n,\mathcal D_2 \subseteq \Sigma^m$, let $f \bullet g = f(g^n)$. The bounded-error quantum query complexity $Q$ satisfies
    \[ Q(f \bullet g) = \Theta(Q(f) \cdot Q(g)).\]
\end{theorem}

 \begin{corollary}[A Direct Sum Theorem]\label{cor:dirsum_theorem}
    For the alphabet set $\Sigma$, functions $g : \mathcal D \to \{0,1\}$ with $\mathcal D \subseteq \Sigma^m$, the bounded-error quantum query complexity $Q$ satisfies
    \[ Q(g^n) = \Theta(nQ(g)).\]
\end{corollary}
\begin{proof}
    Plug $f = \mathrm{id}$ in \thm{composition_theorem}. Then, we only need to prove that $Q(f) = \Omega(n)$. This can be seen, e.g., by reducing to PARITY, which is defined and proved to have $Q(\mathrm{PARITY}) = n/2$ in~\citet{BealsBCMW01}.
\end{proof}

\begin{definition}
We define the following problems:
\begin{itemize}
    \item Decisional Quantum Counting: $f_{n',l,l'} \colon S \to \{0,1\}$ where $S \subset \{0,1\}^{n'},l \ne l'$ is a partial Boolean function defined as
    \begin{equation}
        f_{n',l,l'}(x_0,x_1,\ldots,x_{n'-1})=\left\{\begin{array}{cc}
            0 & \text{if } |X| = l \\
            1 & \text{if } |X| = l' \\
            \text{not defined} & \text{otherwise}
        \end{array}\right.
    \end{equation}
    where $|X| = \sum_{i=0}^{n'-1} x_i$. The notation $f_{n',l,l'}^{(k)} \colon S^k \to \{0,1\}^k$ represents the repeated direct product of $f$, i.e. $f_{n',l,l'}^{(k)}\left(x^{(1)},x^{(2)},\ldots,x^{(k)}\right) = \left(f_{n',l,l'}\left(x^{(1)}\right),f_{n',l,l'}\left(x^{(2)}\right),\ldots,f_{n',l,l'}\left(x^{(k)}\right)\right)$.
    \item Approximate Bits Finding: $\Approx_{k} \subseteq \{0,1\}^k \times \{0,1\}^k$ is a relation problem where for input bits $x_1,\ldots,x_k$, output bits $c_1,\ldots,c_k$ are correct iff the hamming distance between $x$ and $c$ is less than $\varepsilon_0 \cdot k$ for some absolute constant $0<\varepsilon_0<1$ to be determined in the proof of \thm{k-means-lower}.
    \item The operator $\circ$ composites a relation and a function in the natural way, resulting in a relation problem on $S^k \times \{0,1\}^k$.
\end{itemize}
\end{definition}

\begin{theorem}[Theorem 1.3 of \citealt{NayakW99}] \label{thm:lower-quantum-counting}
    Any bounded-error quantum algorithm that computes $f_{n',l,l'}$, given the input as an oracle, must make $\Omega\left(\sqrt{n' / \Delta} + \sqrt{n'(n' - m) / \Delta}\right)$, where $\Delta := |l - l'|$ and $m \in \{l,l'\}$ s.t. $|m - n'/2|$ is maximized. 
\end{theorem}

\subsection{Proof of \thm{lower-mqc-counting-main}}\label{append:proof-lower-mqc-counting-main}
Now, we give the proof of \thm{lower-mqc-counting-main}, which is restated below:

\begin{theorem}[Quantum Lower Bound for  Multidimensional Counting]\label{thm:lower-mqc-counting-append}
    Every quantum algorithm that solves the multidimensional counting problem (Definition \ref{def:mqcounting}) w.p. at least $\frac 2 3$ uses at least $\Omega\left(\sqrt{nk}\varepsilon^{-1/2}\right)$ queries to $O_\tau$.
\end{theorem}

\begin{proof}
    Let $T = \lfloor 0.1 \varepsilon^{-1} \rfloor$. Assume $M$ is even. We reduce from the problem $f_{n/m,T,T+1}^{m/2}$. By \thm{lower-quantum-counting} and \cor{dirsum_theorem}, $Q\left(f_{n/m,T,T+1}^m\right) = \Omega\left(\sqrt{nm}\varepsilon^{-1/2}\right)$. The reduction applies by defining $\tau_i = 2a + x_i + 1$ for $i = am / 2 + b$ where $1 \le b \le m/2$ , calling the quantum multidimensional counter to get $n_1,n_2,\ldots,n_m$, and outputting $n_2 - T, n_4 - T, \ldots, n_m - T$.
\end{proof}

\subsection{Proof of \thm{lower-clustering-main}}\label{append:proof-lower-clustering-main}

Now, we prove \thm{lower-clustering-main}, which is restated below:
\begin{theorem}[Quantum Lower Bounds for $k$-means and k-median]\label{thm:lower-clustering-append}
Assume that $\varepsilon$ is sufficiently small. Consider the Euclidean $k$-means/median problem on data set $D = \{x_1,\ldots,x_n\} \subset \mathbb R^d$. Assume a quantum oracle $O_x \ket{i, b} := \ket{i, b \oplus x_i}$. Then, every quantum algorithm outputs the followings with probability $2/3$ must have quantum query complexity lower bounds for the following problems:
\begin{itemize}[leftmargin=*]
     \item An $\varepsilon$-coreset: $\Omega\left(\sqrt{nk}\varepsilon^{-1/2}\right)$ for $k$-means and $k$-median (\thm{k-means-coreset-lower});
     
     \item An $\varepsilon$-estimation to the value of the objective function: $\Omega\left(\sqrt{nk}+\sqrt n \varepsilon^{-1/2}\right)$ for $k$-means and $k$-median (\thm{k-means-cost-lower});
    
    \item A center set $C$ such that $\cost(C) \le (1 + \varepsilon)\cost\left(C^*\right)$ where $C^*$ is the optimal solution: $\Omega\left(\sqrt{nk}\varepsilon^{-1/6}\right)$ for $k$-means; $\Omega\left(\sqrt{nk}\varepsilon^{-1/3}\right)$ for $k$-median (\thm{k-means-lower}). 
\end{itemize}
\end{theorem}

We prove these different settings separately as follows.

\subsubsection{Coreset Output}
\begin{theorem}\label{thm:k-means-coreset-lower}
    Assume that $\varepsilon$ is sufficiently small. Consider the Euclidean $(k,z)$-clustering problem on data set $D = \{x_1,\ldots,x_n\} \subset \mathbb R^d$. An oracle $O_x \ket{i, b} := \ket{i, b \oplus x_i}$ is accessible. Then, every quantum algorithm that outputs an $\varepsilon$-coreset w.p. at least $\frac 2 3$ uses at least $\Omega\left(\sqrt{nk}\varepsilon^{-1/2}\right)$ queries to $O_x$.
\end{theorem}

\begin{proof}
We reduce from the multidimensional counting problem. The instance is 1-dimensional. Let $B=n^{100}$. The reduction is simply defining $x_i = \tau_i \cdot B$. After getting an $\varepsilon$-coreset from the $(m,z)$-clustering solver, we are able to query an $\varepsilon$-estimate of $\cost_z(D, C_i)$ where $C_i = \{B, 2B, \ldots, i \cdot B + 1, mB\}$, which equals to $|D_i|$, completing the proof.
\end{proof}

\subsubsection{Objective Function Estimation}
\begin{theorem}\label{thm:k-means-cost-lower}
    Assume that $\varepsilon$ is sufficiently small. Consider the Euclidean $(k,z)$-clustering problem on data set $D = \{x_1,\ldots,x_n\} \subset \mathbb R^d$. An oracle $O_x \ket{i, b} := \ket{i, b \oplus x_i}$ is accessible. Then, every quantum algorithm that outputs a real number $\tilde A \in (1 \pm \varepsilon) \min_{C^*} \cost(C^*)$ w.p. at least $\frac 2 3$ uses at least $\Omega\left(\sqrt{nk}+\sqrt n \varepsilon^{-1/2}\right)$ queries to $O_x$.
\end{theorem}

\begin{proof}
Our proof has two parts: $\Omega(\sqrt{nk})$ and $\Omega(\sqrt n \varepsilon^{-1/2})$.

First, we prove that the complexity is $\Omega(\sqrt{nk})$. We can assume that $k = o(n)$. We reduce from the problem $f_{n, k, k + 1}$, which has lower bound $\Omega(\sqrt{nk})$ by \thm{lower-quantum-counting}. The instance is one-dimensional. We set $x'_i = 0$ if $x_i = 0$ and $x'_i = i$ if $x_i=1$. Then, we estimate the $(k,z)$-clustering objective function. In the 0-case, the objective function value must be 0 (we have $k$ centers for $k$ points s.t. $x_i=1$); in the 1-case, the objective function value is greater than 0. Thus, an $\varepsilon$-approximation is able to distinguish two cases, completing this part.

Second, we prove that the complexity is $\Omega(\sqrt{n}\varepsilon^{-1/2})$, even for $k=1$. Let $T = \lfloor 0.1 \varepsilon^{-1} \rfloor$. We may assume $\varepsilon^{-1} = o(n)$. We reduce from the $f_{n,T,T+1}$. The instance is one-dimensional. The reduction is simply setting $x'_i = x_i$. Let $a = \sum_{i=1}^n x_i$. We need to distinguish two cases: $a = T$ and $a = T + 1$. The objective function is $f(x) = a|x|^z + (n - a)|1 - x|^z$. Assume $z > 1$. By calculus, we can see $\min_x f(x) = \frac{a(n-a)}{\left((n-a)^{1/(z-1)} + a^{1/(z-1)}\right)^{z-1}}$. Thus, to prove that an $\varepsilon$-estimation can distinguish two cases, we must prove that, 
$$\frac{T(n-T)}{\left((n-T)^{1/(z-1)} + T^{1/(z-1)}\right)^{z-1}} < (1 - \varepsilon)\frac{(T+1)(n-T-1)}{\left((n-T-1)^{1/(z-1)} + (T+1)^{1/(z-1)}\right)^{z-1}}.$$ 
Because $T = o(n)$, $\frac{\mathrm{RHS}}{\mathrm{LHS}} \sim (1-\varepsilon) \frac{T+1}{T} > 1$ for sufficiently small $\varepsilon$. As for $z = 1$, $f(x) = a$ so the above argument is also valid.
\end{proof}

\subsubsection{Center Set Output}
\begin{theorem}\label{thm:k-means-lower}
    Assume that $\varepsilon$ is sufficiently small. Consider the Euclidean $(k,z)$-clustering problem on data set $D = \{x_1,\ldots,x_n\} \subset \mathbb R^d$. An oracle $O_x \ket{i, b} := \ket{i, b \oplus x_i}$ is accessible where the second register saves the binary representation of a real number and can have any polynomial number of qubits. Then, every quantum algorithm that outputs the optimal centers $C = \{c_1^,\ldots,c_k\}$ such that $\cost(C) \le (1 + \varepsilon) \min_{C^*} \cost(C^*)$ w.p. at least $\frac 2 3$ uses at least $\Omega\left(\min\left(n, \sqrt{nk}\varepsilon^{-1/3z}\right)\right)$ queries to $O_x$.
\end{theorem}

\begin{proof}
For convenience, we assume that $n$ is a multiple of $k$ and focus on the $z = 2$ case (the proof for the $z \ne 2$ case is similar). We only need to prove when $\frac{1}{\varepsilon^{-1/3}} = o(k)$. Let $T = \left\lfloor (16\varepsilon_0\varepsilon)^{-1/3} \right\rfloor-1$ where $\varepsilon_0$ is to be defined. We will reduce an instance of the $2k$-means problem from the problem $\mathcal P = \Approx_{k} \circ f_{n/k,T,T+1}^{(k)}$.

Intuitively, solving the problem $\mathcal P$ need to solve $k$ independent cases of a quantum counting problem (distinguishing $l$ 1s from $l'$ 1s), but only need to be correct on a constant fraction of instances. We prove that the quantum algorithm must cost $k$ times the queries of the $k=1$ case, which can be lower bounded by \thm{lower-quantum-counting}.

Now we describe the reduction. The hard instance is consist of $k$ unit balls far apart and each unit ball has $T$ or $T+1$ unit vectors and $n-T$ or $n-T+1$ origins on it. Let $\mathcal A$ be an optimal algorithm for the $2k$-means problem. The input of the problem $\mathcal P$ is $x \in S^k$ where $S = \{x \in \{0,1\}^{n/k} : |X| = T \text{ or } T+1\}$. Let $x = x_1x_2\ldots x_n$. We map each $x_i$ to a point in $\mathbb R^d$ for $d = 100\log k / \varepsilon^2$. Let $B = n^{100}$ and $i = ak + b$ for $0 \le a < k,1 \le b \le n/k$. Define $v_i = (i \cdot B, 0, \ldots, 0)$. If $x_i = 0$, we map it to the point $v_{2a}$; if $x_i = 1$, we draw $t_i \gets \{2, 3, \ldots, d\}$ uniformly and randomly. Then we set the point as $v_{2a+1} + e_{t_i}$ where $\{e_1,e_2,\ldots,e_d\}$ is the standard basis of $\mathbb R^d$. Note that by the union bound and the Birthday Paradox, $\{t_i\}$ are distinct in each group of size $n/k$ w.h.p. We condition on this from now on. Since the reduction is classical, simple, and local, we can implement the oracle for the $2k$-means problem directly. Finally, we call $\mathcal A$ on the above instance and output $(\round(t(\| c_2 - v_1 \|_2)), \ldots, \round(t(\| c_{2i} - v_{2i-1} \|_2)), \ldots, \round(t(\| c_{2n} - v_{2n-1} \|_2)))$, where:

\begin{itemize}
    \item $t(x) = \frac{1}{\frac{1}{\sqrt{T}} - \frac{1}{\sqrt{T+1}}} \left[\frac{1}{\sqrt{T}}-\min(\max(x, \frac{1}{\sqrt{T+1}}), \frac{1}{\sqrt{T}})\right]$ is a normalizing mapping;\\
    \item $\round(x) = \left\{\begin{array}{cc}
        0 & x \le \frac 1 2 \\
        1 & x > \frac 1 2
    \end{array}\right.$. 
\end{itemize}

By easy adjustments (we can assume that there is a point in each of $2k$ balls centering at $v_i$ with radius $B/5$ because otherwise the cost is larger than $B/10$, which is very large; then, we can move each center to its nearest point on the unit ball centered at $v_i$), we can assume without loss of generality that the solutions outputted by the $2k$-means solver have the following properties:
\begin{enumerate}
    \item $c_{2a-1} = v_{2a-2}$ for $1 \le a \le k$;
    \item $\|c_{2a}, v_{2a-1}\|_2 \le 1$.
\end{enumerate}
Then, the cost of the clustering can be seen as the sum of costs of $k$ independent $1$-mean problems and each of the problem has size $T$ or $T+1$. It is well known that, in the $1$-mean problem of size $n$, $\cost(c) = \cost(c^*) + n\| c - c^* \|^2$ where $c^* = \frac{x_1+x_2+\cdots+x_n}{n}$ is the optimal center. We decompose $c_{2a} = v_{2a-1} + g_{a}$. Let $T_a = \{e_{t_i} : x_i = 1, (a-1)k < i \le ak \}$. (Recall that $|T_a| = T$ or $T+1$.) Define $g_a^* = \frac{1}{|T_a|}\sum_{e \in T_a} e$. Then, the clustering is $\varepsilon$-optimal if and only if
\begin{align}
    &\sum_{a=1}^k |T_a| \lVert g_a - g_a^* \rVert^2 < \varepsilon \cdot \sum_{a=1}^k \frac{1}{2|T_a|}\sum_{x \ne y \in T_a} \dist(x,y)^2\\
    \Longrightarrow\quad  &   \sum_{a=1}^k |T_a| \left(\| g_a \| - \|g_a^*\|\right)^2 < \varepsilon \cdot \sum_{a=1}^k (|T_a| - 1)\\
    \Longrightarrow\quad & T \cdot \sum_{a=1}^k \left(\| g_a \| - \frac{1}{\sqrt{|T_a|}}\right)^2 < \varepsilon \cdot k \cdot (T - 1)\\
    \Longrightarrow\quad & \sum_{a=1}^k \left(\| g_a \| - \frac{1}{\sqrt{|T_a|}}\right)^2 < \varepsilon \cdot k.
\end{align}

Note that $t\left(\frac{1}{\sqrt{|T_a|}}\right)$ gives $f_{n/k,T,T+1}$ in the $i$-th block of size $n/k$, justifying the inner function of the problem $\mathcal P$.

\begin{align}
    \Longrightarrow\quad & \left(\frac{1}{\sqrt{T}} - \frac{1}{\sqrt{T+1}}\right)^2\sum_{a=1}^k \left(t(\| g_a \|) - t\left(\frac{1}{\sqrt{|T_a|}}\right)\right)^2 < \varepsilon \cdot k \\
    \Longrightarrow\quad & \frac{1}{4T \cdot (T + 1)^2} \sum_{a=1}^k \left(t(\| g_a \|) - t\left(\frac{1}{\sqrt{|T_a|}}\right)\right)^2 < \frac{\varepsilon_0}{16(T+1)^3} \cdot k \\
    \Longrightarrow\quad & \sum_{a=1}^k \left(t(\| g_a \|) - t\left(\frac{1}{\sqrt{|T_a|}}\right)\right)^2 < \frac{\varepsilon_0}{4} \cdot k. 
\end{align}

Intuitively, the $k$-means solver needs to solve $k$ independent cases of $f_{n/k,T,T+1}$ where the right answers are $t\left(|T_a|^{-1/2}\right)$. $t(\| g_a \|) \in [0,1]$ are a fractional guess in $[0,1]^k$ of $\{0,1\}^k$. For convenience, we round the output. A simple lemma is required to bound the error of rounding:

\begin{lemma}
  Given $x_1, x_2, \ldots, x_k \in [0,1]$ and $y_1, y_2, \ldots, y_k \in \{0,1\}$, we have that
  \[ \sum_{i=1}^k \bm 1_{\round(x_i) \ne y_i} \le 4\sum_{i=1}^k (x_i - y_i)^2 \].
\end{lemma}
\begin{proof}
    One has:
    \begin{align*}
        & \sum_{i=1}^k (x_i - y_i)^2 \\
        = & \sum_{\round(x_i) \ne y_i} (x_i - y_i)^2 + \sum_{\round(x_i) = y_i} (x_i - y_i)^2\\
        \ge & \frac 1 4 \sum_{t(x_i) \ne y_i} 1. \qedhere
    \end{align*}
\end{proof}

Back to the proof of \thm{k-means-lower}. Plugging $x_i = t(\| g_a \|)$ and $y_i = t\left(\frac{1}{\sqrt{|T_a|}}\right)$ in the above lemma, we have that the hamming distance between the output of our solver for $\mathcal P$ and $f_{n/k,T,T+1}^{(k)}$ is less than $\varepsilon_0 k$, so it is indeed a solver for the relation problem $\mathcal P$.

Now it is sufficiently to prove the lower bound for the problem $\mathcal P$. Consider another problem defined simply as $\mathcal P' = f_{n/k,T,T+1}^{(k)}$. Applying \cor{dirsum_theorem}, we have that $Q\left(\mathcal P'\right) = k Q(f_{n/k,T,T+1})$. By \thm{lower-quantum-counting}, $Q(f_{n/k,T,T+1}) = \Theta\left(\sqrt{n/k T}\right)$. Hence, $Q\left(\mathcal P'\right) = \Theta\left(\sqrt{nk}\epsilon^{-1/6}\right)$.

Thus, there exists a constant $C_0 > 0$ such that every quantum algorithm that solves $\mathcal P'$ w.p. at least $\frac 2 3$ uses at least $C_0 \sqrt{nk}\varepsilon^{-1/6}$ queries. Let $\mathcal A$ query the oracle for $t$ times. We construct an algorithm $\mathcal A'(x)$ for $\mathcal P'$ from $\mathcal A$: $\mathcal A'$ calls $c \leftarrow \mathcal A(x)$ and then uses the Grover search (\lem{getall}) to find the set $S = \{i \in [k] : f_{n/k,T,T+1}^{(k)}(x)_i \ne c_i \}$ and then flips the bits of $c$ in $S$ and output $c$. By the definition of $\mathcal P'$, $|S| \le \varepsilon_0 \cdot k$. And $f_{n/k,T,T+1}^{(k)}(x)_i$ can be computed by $\frac{\pi}{2}\sqrt{nT/k} + n/k \le 2\sqrt{nT/k}$ queries by the Grover search too. Thus, finding $S$ costs $2(\frac{\pi}{2}\sqrt{\varepsilon_0 \cdot k \cdot k} + \varepsilon_0 \cdot k)\sqrt{nT/k} \le 100 \varepsilon_0^{5/6} \sqrt{nk }\varepsilon^{-1/6}$ queries. We then have $t + 100 \varepsilon_0^{5/6} \sqrt{nk }\varepsilon^{-1/6} \ge C_0 \sqrt{nk}\varepsilon^{-1/6}$. Now set $\varepsilon_0 = \left(\frac{C_0}{1000}\right)^{6/5}$ and solve the inequality, we get $t \ge 0.9 C_0 \sqrt{nk}\varepsilon^{-1/6}$, completing the proof.
\end{proof}


\end{document}